\documentclass[jair,twoside,11pt,theapa]{article}
\usepackage{jair, theapa, rawfonts}

\ShortHeadings{Practical Parallel Algorithms for Submodular Maximization}
{Cui, Han, Tang, Li, Zhiyuli \& Li}
\firstpageno{1}

\usepackage{threeparttable}
\usepackage[ruled,vlined,linesnumbered]{algorithm2e}
\usepackage{amsmath}
\usepackage{amssymb}
\usepackage{enumitem}
\usepackage{bbm}
\usepackage{amsmath,amsthm,amsfonts,amssymb,bm}
\usepackage{xspace}
\usepackage{graphicx}
\usepackage{enumitem}
\usepackage{booktabs}
\usepackage[table]{xcolor}
\usepackage{soul}

\newcommand{\bl}[1]{{\color{black} #1}}

\newcommand{\blt}[1]{{\color{black} #1}}

\newcommand{\cG}{\ensuremath{\mathcal{G}}\xspace}
\newcommand{\bP}{\ensuremath{\mathbb{P}}}
\newcommand{\TE}{\textsc{ThreshSeq}\xspace}

\newcommand{\I}{\mathcal{I}}
\newcommand{\N}{\mathcal{N}}
\newcommand{\OO}{\mathcal{O}}
\newcommand{\E}{\mathbb{E}}

\newcommand{\U}{\mathcal{U}}
\newcommand{\LL}{\mathcal{L}}

\newcommand{\OPT}{\mathrm{OPT}}

\newcommand{\RB}{\textsf{RandBatch}}
\newcommand{\RQ}{\textsf{GetSEQ}}
\newcommand{\PS}{\textsf{ParSSP}}
\newcommand{\algone}{\textsf{ParSKP}}
\newcommand{\algtwo}{\textsf{\PS}}
\newcommand{\onerun}{\textsf{Probe}}

\newtheorem{definition}{Definition}
\newtheorem{theorem}{Theorem}
\newtheorem{lemma}{Lemma}

\begin{document}

\title{Practical Parallel Algorithms for Non-Monotone Submodular Maximization}

\author{\name Shuang Cui 
	\email scui@suda.edu.cn \\
       \addr School of Computer Science and Technology,\\ Soochow University,\\ Suzhou, Jiangsu, China
       \AND
       \name Kai Han \rm{(Corresponding Author)} \email hankai@suda.edu.cn \\ 
       \addr School of Computer Science and Technology,\\ Soochow University,\\ Suzhou, Jiangsu, China
       \AND
       \name Jing Tang \email jingtang@ust.hk \\
       \addr The Hong Kong University of Science and Technology (Guangzhou),\\ The Hong Kong University of Science and Technology,\\ China
    \AND
   \name Xueying Li 
   \email xiaoming.lxy@alibaba-inc.com \\
    \name Aakas Zhiyuli
   \email aakas.lzy@alibaba-inc.com\\
   \addr Alibaba Group,\\ 
   Hangzhou, Zhejiang, China
   \AND Hanxiao Li
   \email 2023060204@mails.qust.edu.cn\\
   \addr Qingdao University of Science and Technology,\\
   Qingdao, Shandong, China
}


\maketitle

\begin{abstract}
Submodular maximization has found extensive applications in various domains within the field of artificial intelligence, including but not limited to machine learning, computer vision, and natural language processing. With the increasing size of datasets in these domains, there is a pressing need to develop efficient and parallelizable algorithms for submodular maximization. One measure of the parallelizability of a submodular maximization algorithm is its adaptive complexity, which indicates the number of sequential rounds where a polynomial number of queries to the objective function can be executed in parallel. In this paper, we study the problem of non-monotone submodular maximization subject to a knapsack constraint, and propose a low-adaptivity algorithm achieving an $(1/8-\epsilon)$-approximation with practical $\tilde{\OO}(n)$ query complexity. Moreover, we also propose the first algorithm with both provable approximation ratio and sublinear adaptive complexity for the problem of non-monotone submodular maximization subject to a $k$-system constraint. As a by-product, we show that our two algorithms can also be applied to the special case of submodular maximization subject to a cardinality constraint, and achieve performance bounds comparable with those of state-of-the-art algorithms. Finally, the effectiveness of our algorithms is demonstrated by extensive experiments on real-world applications.
\end{abstract}

\section{Introduction}\label{sc:intro}
Submodular maximization algorithms have played a critical role in advancing the field of artificial intelligence, particularly in the areas of machine learning (e.g., non-parametric learning~\shortcite{qian2019maximizing,kulesza2012determinantal,lawrence2002fast}, active learning~\shortcite{golovin2010adaptive,wei2015submodularity,golovin2011adaptive}, data summarization~\shortcite{balkanski2018non,mirzasoleiman2018streaming,mitrovic2018data,amanatidis2022fast}), computer vision (e.g., object detection~\shortcite{angelova2013efficient,zhu2014submodular}, image segmentation~\shortcite{kim2011distributed}), and natural language processing (e.g., text classification~\shortcite{lei2019discrete}, document summarization~\shortcite{lin2012learning,kulesza2012determinantal}). As a result, submodular maximization has been widely studied under various constraints such as cardinality, knapsack, matroid, and $k$-system constraints. Many algorithms in this area adopt the greedy search strategy~(e.g., continuous greedy algorithms in \shortcite{calinescu2011maximizing}), but may have large \textit{query complexity} to achieve a good approximation ratio, where query complexity refers to the number of evaluations to the objective function. In practice, evaluating the objective function may be time-consuming~\shortcite{dueck2007non,das2008algorithms,kazemi2018scalable}, and this situation is further exacerbated by the proliferation of ``big data'', for which simply reducing query complexity is often insufficient to get efficient algorithms.
Thus, parallelization has received increased attention for submodular maximization.

Unfortunately, traditional greedy algorithms for submodular maximization are inherently sequential and adaptive, which makes them unsuitable for being parallelized. Some efforts have been devoted to designing distributed submodular maximization algorithms using parallel models such as MapReduce~\shortcite{mirzasoleiman2013distributed,kumar2013fast,barbosa2015power,mirzasoleiman2016distributed,epasto2017bicriteria,kazemi2021regularized}, but these algorithms can still be highly adaptive, as they usually run sequential greedy algorithms on each of the machines. Recently, \shortcite{balkanski2018adaptive} proposed submodular maximization algorithms with low \textit{adaptive complexity} (a.k.a.\ ``adaptivity''), where only a sub-linear number of \textit{adaptive rounds} are incurred and polynomially-many queries can be executed in parallel in each adaptive round. Subsequently, a lot of studies have appeared to design low-adaptivity algorithms; many of them concentrate on the submodular maximization with a cardinality constraint (\textbf{SMC}) problem~(e.g., \shortcite{kazemi2019submodular,fahrbach2019submodular,balkanski2019exponential}).

Besides the SMC, one of the most fundamental problems in submodular optimizations is the problem of submodular maximization subject to a knapsack constraint (\textbf{SKP}), which has many applications both for monotone and non-monotone submodular functions~\shortcite{kulik2009maximizing,lee2010maximizing,badanidiyuru2014fast}. Surprisingly, although the SKP has been extensively studied since the 1980s~\shortcite{wolsey1982maximising}, there exist only few studies on designing low-adaptivity algorithms for it. In particular, \shortcite{chekuri2019submodular} provides a $(1-1/e-\epsilon)$-approximation in $\mathcal{O}(\log{n})$ adaptive rounds for monotone SKP, while \shortcite{ene2019submodular} presents a $(e-\epsilon)$-approximation in $\mathcal{O}(\log^2 n)$ adaptive rounds for non-monotone SKP. However, both \shortcite{chekuri2019submodular} and \shortcite{ene2019submodular} assume oracle access to the multilinear extension of a submodular function and its gradient, which incurs high query complexity for estimating multilinear extensions accurately, making their algorithms impractical in real applications~\shortcite{amanatidis2021submodular}. Very recently, \shortcite{amanatidis2021submodular} study the non-monotone SKP problem and present a parallelizable algorithm dubbed \textsf{ParKnapsack} with $\tilde{\OO}(n)$\footnote{Throughout the paper, we use the $\tilde{O}$ notation to suppress poly-logarithmic factors.} query complexity and an approximation ratio of $\frac{-3 \epsilon (\epsilon (8\epsilon-2 \sqrt{3}-13)+4 \sqrt{3}+2)+4 \sqrt{3}-6}{3 ((2 \sqrt{3}-3) \epsilon^3+(11-6 \sqrt{3}) \epsilon+4 (\sqrt{3}-1))}=(3 - \sqrt{3})/12-\Theta(\epsilon)\approx 0.106-\Theta(\epsilon)$. Unfortunately, as pointed out by our conference version of this paper~\shortcite{cui2023practical}, the performance analysis of ~\shortcite{amanatidis2021submodular} contains critical errors and hence their performance bounds do not hold. Thus, our work is the first to provide a low-adaptivity algorithm that can achieve a provable approximation ratio with practical query complexity for the
non-monotone SKP. After the publication of our conference version~\shortcite{cui2023practical,cui2023arxiv}, \shortcite{amanatidis2023submodular} (which is the full version of \shortcite{amanatidis2021submodular}) provided a revised algorithm achieving the same $0.106-\Theta(\epsilon)$ ratio as that in \shortcite{amanatidis2021submodular}, but this ratio is still worse than the $0.125-\epsilon$ ratio provided by this paper.

Another fundamental problem in submodular optimizations is the problem of submodular maximization subject to a $k$-system constraint (\textbf{SSP}), as $k$-system constraints capture a wide variety of constraints, including matroid constraints, intersection of
matroids, spanning trees, graph matchings, scheduling, and even planar subgraphs. Since the 1970s, significant efforts have been devoted to solving the SSP. The state-of-the-art approximation ratios are $1/(k+1)$~\shortcite{fisher1978analysis} and $(1+\epsilon)^{-1}(\sqrt{k}+1)^{-2}$~\shortcite{cui2021randomized} for monotone submodular functions and non-monotone submodular functions, respectively. Recently, \shortcite{quinzan2021adaptive} proposes the first low-adaptivity algorithm for the non-monotone SSP. Their algorithm can achieve an approximation ratio of $(1+\epsilon)^{-1}{(1-\epsilon)^{2}}(k+2\sqrt{2(k+1)}$
$+5)^{-1}$ under the adaptive complexity of $\OO(\sqrt{k}\log^2 n\log\frac{n}{k})$ and the query complexity of $\tilde{\OO}(\sqrt{k}n)$. Regrettably, as demonstrated in~Appendix \ref{app:kuhnle_error}, their analysis contains serious errors that invalidated their algorithm’s approximation ratio. Therefore, it remains an open question whether there exists a low-adaptivity algorithm that can achieve a provable approximation ratio for the non-monotone SSP.\\
\textbf{Contributions.}~In this paper, we present two practical low-adaptivity algorithms named \algone~and \PS~that provide confirmative answers to the open problem mentioned above and significantly improve the existing results. Our contributions include:
\begin{itemize}
	\item For the non-monotone SKP, we propose \algone~algorithm that uses $\OO(\log^2 n)$ adaptivity and practical $\tilde{\OO}(n)$ query complexity (or uses near-optimal $\OO(\log n)$ adaptivity and $\tilde{\OO}(n^2)$ query complexity) to achieve an approximation ratio of $1/8-\epsilon$ which is best among the existing practical parallel algorithms for the SKP problem.
	\item For the non-monotone SSP, we propose \PS~algorithm that uses $\OO(\sqrt{k}\log^3 n)$ adaptivity and practical $\tilde{\OO}(\sqrt{k}n)$ query complexity (or uses $\OO(\sqrt{k}{\log^2}n)$ adaptivity and $\tilde{\OO}(\sqrt{k}n^2)$ query complexity) to achieve an approximation ratio of $(1-\epsilon)^{5}(\sqrt{k+1}+1)^{-2}$. To the best of our knowledge, our \PS~is the first low-adaptivity algorithm that can achieve a provable approximation ratio for the non-monotone SSP.
	\item As a by-product, our two algorithms can be directly used to address the non-monotone SMC, where the approximation ratio of \PS~can be tightened to $1/4-\epsilon$ under the same complexities listed above. 
 That is to say, our \PS~algorithm achieves the best approximation ratio among existing practical low-adaptivity algorithms for the SMC problem.
	\item We conduct extensive experiments using several applications including revenue maximization, movie recommendation and image summarization. The experimental results demonstrate that our algorithms can achieve comparable utility using significantly fewer adaptive rounds than existing non-parallel algorithms, while also achieving significantly better utility than existing low-adaptivity algorithms.
\end{itemize}
\textbf{Challenges and Techniques.} Naively, it seems that slightly adjusting the low-adaptivity algorithms for the monotone submodular maximization problems can address their non-monotone variants. However, we find that the techniques developed under the monotone setting heavily rely on monotonicity, which may incur subtle issues when applying to the non-monotone scenarios. For instance, as pointed out in~Appendix \ref{app:kuhnle_error} and \shortcite{Chen2022}, the approximation claims in some recent work~\shortcite{fahrbach2019non,quinzan2021adaptive} are flawed by the abuse of existing techniques developed for monotone functions. Moreover, \shortcite{amanatidis2021submodular} proposes the \textsf{ParKnapsack} algorithm sophisticatedly tailored to non-montone functions, which again gives flawed analysis on approximation guarantees as pointed in our conference version. As can be seen, designing approximation parallel algorithms for non-monotone SKP/SSP achieving both low-adaptivity and low-complexity is challenging. To overcome this challenge, our \algone~algorithm adopts \bl{a novel method} using a more sophisticated filtering procedure dubbed \onerun~where multiple solutions are created in different ways given an ideal threshold, which can provide theoretical guarantees correctly. Moreover, our \PS~algorithm adopts another method where there is no need to guess an ideal threshold, but we use decreasing thresholds and introduce a ``random batch selection'' operation to bypass the difficulties caused by non-monotonicity.

This manuscript provides an expanded and refined presentation of the research that was initially introduced in a preliminary conference paper at AAAI-2023~\shortcite{cui2023practical}. 
The main revisions can be summarized as follows:
\begin{itemize}
	\item We present the \PS~algorithm (Section~\ref{sec:ssp}), which is a modification of an algorithm previously introduced in the conference version. Our \PS~algorithm is the first to achieve both a provable approximation ratio and sublinear adaptive complexity for the non-monotone SSP. 
	Moreover, we perform a set of comparative experiments to validate the performance of our \PS~algorithm for this problem (Section~\ref{sec:exp-ssp}).
	\item We have added a set of comparative experiments on a small dataset to compare our \algone~algorithm with a parallel algorithm that has a theoretically optimal approximation ratio but impractical query complexity (Section~\ref{sec:exp-skp2}). Additionally, we also add experiments for the SMC problem (Section \ref{sec:exp-smc}). These experiments further demonstrate the superiority of our algorithm in terms of efficiency and effectiveness.
	\item Due to page limitations, the complete theoretical analysis of our algorithms and the discussions of theoretical analysis errors in \shortcite{quinzan2021adaptive,amanatidis2021submodular} were not included in the conference version. However, we provide all of the complete theoretical analysis (Sections~\ref{sec:pt}-\ref{sec:smc}) and discussion mentioned above in this journal version (Appendix \ref{app:amanatidis_error}-\ref{app:kuhnle_error}).
	\item We have restructured all sections to provide a clearer and more comprehensive explanation of our research objectives and findings.
\end{itemize}

The rest of our paper is organized as follows. The related studies are reviewed in Section~\ref{sec:rw}
and our problem definitions are presented in Section \ref{sc:pre}. Our proposed approximation algorithms and their theoretical analysis are introduced in Section~\ref{sec:pt}. In Section~\ref{sec:smc}, we analyze the performance bounds of our algorithms when applied to the SMC. We present the experimental results in Section~\ref{sec:experiment} before concluding our work in Section \ref{sec:con}. Discussion on the theoretical analysis error in relevant literature has been relegated to Appendix \ref{app:amanatidis_error}-\ref{app:kuhnle_error}.
\begin{table*}[t]
	\caption{Low-adaptivity algorithms for non-monotone submodular maximization.}
	\label{tb:rw}
	\begin{threeparttable}
		\begin{center}
			\resizebox{\textwidth}{!}{
				\begin{tabular}{lclll}
					\toprule
					Constraint & Reference & Ratio & Adaptivity & Queries \\
					\midrule
					Knapsack&\shortcite{ene2019submodular}& $\bm{1/e-\epsilon}$& $\mathcal{O}(\log^2 n)$&\bl{$\tilde{\OO}(n^3)$}\\
					&\shortcite{amanatidis2023submodular}& $(3 - \sqrt{3})/12-\Theta(\epsilon)$& $\bm{\mathcal{O}(\log n)}\mid\mid \mathcal{O}(\log^2 n)$&$\tilde{\OO}(n^2)\mid\mid\bm{\tilde{\OO}(n)}$ \\	
					&\algone (Alg. \ref{alg:skp1})& $1/8-\epsilon$& $\bm{\mathcal{O}(\log n)}\mid\mid \bl{\mathcal{O}(\log^2 n)}$&$\bl{\tilde{\OO}(n^2)}\mid\mid\bm{\tilde{\OO}(n)}$ \\	
					\midrule
					$k$-System 
					&\shortcite{quinzan2021adaptive} & ${(1+\epsilon)^{-1}{(1-\epsilon)^{2}}(k+2\sqrt{2(k+1)}+5)^{-1}}^\#$& $\bl{\mathcal{O}(\sqrt{k}\log^2 n\log\frac{n}{k})}$& $\bm{\tilde{\OO}(\sqrt{k}n)}$\\	
					&\PS(Alg. \ref{alg:ssp}) & $\bm{{(1-\epsilon)^{5}{(\sqrt{k+1}+1)^{-2}}}}$& $\bm{\mathcal{O}(\sqrt{k}{\log^2}n)}\mid\mid \bl{\OO(\sqrt{k}\log^3 n)}$& $\bl{\OO(\sqrt{k}n^2)}\mid\mid\bm{\tilde{\OO}(\sqrt{k}n)}$\\		
					\midrule
					Cardinality
					&\shortcite{chekuri2019parallelizing}& $3-2\sqrt{2}-\epsilon$& $\mathcal{O}(\log^2 n)$&$\tilde{\OO}(nr^4)$\\
					&\shortcite{balkanski2018non}& $1/2e-\epsilon$& $\mathcal{O}(\log^2 n)$&$\tilde{\OO}(nr^2)$\\
					&\shortcite{ene2020parallel}& $\bm{1/e-\epsilon}$& $\bm{\mathcal{O}(\log n)}$&$\tilde{\OO}(nr^2)$\\
					&\shortcite{Chen2022}& $1/6-\epsilon\mid\mid 0.193-\epsilon$& $\bm{\mathcal{O}(\log n)}\mid\mid \mathcal{O}(\log n\log r)$&$\bm{\tilde{\OO}(n)}$ \\		
					&\algone (Alg. \ref{alg:skp1})& $1/8-\epsilon$& $\bm{\mathcal{O}(\log n)}\mid\mid \mathcal{O}(\log n\log r)$&$\tilde{\OO}(nr)\mid\mid\bm{\tilde{\OO}(n)}$ \\	
					&\PS (Alg. \ref{alg:ssp})& $1/4-\epsilon$& $\mathcal{O}(\log^2 n)\mid\mid \mathcal{O}(\log^2n\log r)$&$\tilde{\OO}(nr)\mid\mid \bm{\tilde{\OO}(n)}$ \\
					\bottomrule
				\end{tabular}
			}
		\end{center}
		\begin{tablenotes}
			\footnotesize
			\item[1] Bold font indicates the best result(s) in each setting, $r$ is the largest cardinality of any feasible solution. 
			\item[\#] The approximation ratio is derived from flawed analysis as explained in Appendix \ref{app:kuhnle_error}.
		\end{tablenotes}
	\end{threeparttable}
\end{table*}
\section{Related Work}\label{sec:rw}
In the following, we review several lines of related studies, and list the performance bounds of some representative ones in Table~\ref{tb:rw}.
\subsection{Algorithms for the SKP}\label{sc:rw-skp}
The traditional SKP has been extensively studied, both for monotone and non-monotone submodular functions~(e.g., \shortcite{sviridenko2004note,kulik2013approximations,gupta2010constrained,ene2019nearly,yaroslavtsev2020bring}). For non-monotone SKP, \shortcite{buchbinder2019constrained} achieves the best-known approximation ratio of $0.385$, while some other studies \shortcite{mirzasoleiman2016fast,amanatidis2022fast,han2021approximation,cui2024deletion} provide faster algorithms with weaker approximation ratios. However, all these algorithms have super-linear adaptive complexity unsuitable for parallelization. In terms of low-adaptivity algorithms for SKP, \shortcite{chekuri2019submodular} provides a $(1-1/e-\epsilon)$-approximation in $\mathcal{O}(\log{n})$ adaptive rounds for monotone SKP. For the non-monotone SKP, \shortcite{ene2019submodular} provides an $(1/e-\epsilon)$-approximation with $\mathcal{O}(\log^2 n)$ adaptivity, based on a continuous optimization approach using multi-linear extension. 
However, an excessively large and impractical number of $\Omega(nr^2\log^2 n)$ function valuations are needed in \shortcite{ene2019submodular} for simulating a query to the multilinear extension of a submodular function or its gradient with sufficient accuracy, as described in~\shortcite{fahrbach2019non}. Observing this, \shortcite{amanatidis2021submodular} provides a combinatorial algorithm that uses $\OO(\log^2 n)$ adaptivity and practical $\tilde{\OO}(n)$ query complexity to achieve an approximation ratio of $\frac{-3 \epsilon (\epsilon (8\epsilon-2 \sqrt{3}-13)+4 \sqrt{3}+2)+4 \sqrt{3}-6}{3 ((2 \sqrt{3}-3) \epsilon^3+(11-6 \sqrt{3}) \epsilon+4 (\sqrt{3}-1))}=(3 - \sqrt{3})/12-\Theta(\epsilon)\approx 0.106-\Theta(\epsilon)$. Unfortunately, their approximation ratio is derived from flawed analysis as explained in our conference version \shortcite{cui2023practical}. The work of \shortcite{amanatidis2023submodular} (which is the full version of \shortcite{amanatidis2021submodular}) has adopted a more complex sampling procedure to achieve an approximation ratio of $0.106-\Theta(\epsilon)$, \blt{but this ratio is still worse than the $0.125-\epsilon$ ratio provided by this paper.}


\subsection{Algorithms for the SSP}
Since the 1970s, the SSP has been a topic of widespread interest among the academic community. Research related to the SSP has included~\shortcite{fisher1978analysis,gupta2010constrained,calinescu2011maximizing,mirzasoleiman2016fast,feldman2017greed,feldman2020simultaneous,cui2021randomized}. Among the existing proposals, \shortcite{cui2021randomized} achieves the best-known approximation ratio of $(1+\epsilon)^{-1}(k+2\sqrt{k}+1)^{-1}$ under $\mathcal{O}(\frac{n}{\epsilon}\log \frac{r}{\epsilon})$ time complexity. However, all the studies mentioned above only provide algorithms with super-linear adaptive complexity unsuitable for parallelization. In terms of parallelizable algorithms, to the best of our knowledge, only \shortcite{quinzan2021adaptive} has proposed a low-adaptivity algorithm for the SSP that achieves an approximation ratio of ${(1+\epsilon)^{-1}{(1-\epsilon)^{2}}(k+2\sqrt{2(k+1)}+5)^{-1}}$ under the adaptive complexity of $\OO(\sqrt{k}\log^2 n\log\frac{n}{k})$ and the query complexity of $\tilde{\OO}(\sqrt{k}n)$. Unfortunately, their approximation ratio is derived from flawed analysis as pointed in our conference version.
\subsection{Algorithms for the SMC}
For the monotone SMC, several parallelizable algorithms have been proposed such as~\shortcite{ene2019bsubmodular,balkanski2019optimal,chekuri2019parallelizing,breuer2020fast,chen2021best}. For the non-monotone SMC, \shortcite{chekuri2019parallelizing} and \shortcite{balkanski2018non} propose parallelizable algorithms with $3-2\sqrt{2}-\epsilon$ and $1/(2e)-\epsilon$ approximation ratios, respectively, both under $\OO(\log^2 n)$ adaptivity, while \shortcite{ene2020parallel} achieve an improved ratio of $1/e-\epsilon$ under $\OO(\log n)$ adaptivity. However, all these studies use multilinear extensions and have high query complexity (larger than $\Omega(nr^2)$). Another study in this line \shortcite{fahrbach2019non} aims to reduce both adaptivity and query complexity, but achieving a relatively large ratio of $0.039-\epsilon$. Subsequently, \shortcite{kuhnle2021nearly} claims a $(0.193-\epsilon)$-approximation with $\OO(\log^2 n)$ adaptivity. However, \shortcite{Chen2022} identifies non-trivial errors in both \shortcite{fahrbach2019non} and \shortcite{kuhnle2021nearly}, and propose a new adaptive algorithm to fix the $(0.193-\epsilon)$-approximation. Despite these efforts, our \PS~algorithm still achieves a superior $(1/4-\epsilon)$-approximation for the non-monotone SMC.
\section{Preliminaries}\label{sc:pre}
We provide the some basic definitions in the following:  

\begin{definition}[Submodular Function, defined by \shortcite{fisher1978analysis}]\label{definition:submodular}
	Given a finite ground set $\N$ with $|\N|=n$,  a set function $f\colon 2^{\N}\mapsto \mathbb{R}$ is submodular if for all sets $X,Y\subseteq \N\colon f(X)+f(Y)\geq f(X\cup Y)+f(X\cap Y)$.
\end{definition}
In this paper, we allow $f(\cdotp)$ to be non-monotone, i.e., $\exists X\subseteq Y\subseteq \N\colon f(X)>f(Y)$, and we assume that $f(\cdotp)$ is non-negative, i.e., $f(X)\geq 0$ for all $X\subseteq \N$.
\begin{definition}[Independence System]\label{definition:independence_system}
	Given a finite ground set $\N$ and a collection of sets $\I\subseteq 2^{\N}$, the pair $(\N,\I)$ is called an independence system if it satisfies: (1) $\emptyset\in\I$; (2) if $X\subseteq Y\subseteq\N$ and $Y\in\mathcal{I}$, then $X\in\mathcal{I}$.
\end{definition}

Given an independence system $(\N,\I)$ and any two sets $X\subseteq Y\subseteq \N$, $X$ is called a \textit{base} of $Y$ if $X\in \I$ and $X\cup \{u\}\notin \I$ for all $u\in Y\setminus X$. A $k$-system is a special independence system defined as:

\begin{definition}[$k$-system]\label{definition:ksystem}
	An independence system $(\N,\I)$ is called a $k$-system ($k\geq1$) if $|X_1|\leq k|X_2|$ holds for any two bases $X_1$ and $X_2$ of any set $Y\subseteq \N$.
\end{definition}
We assume that each element $u\in \N$ has a cost $c(u)>0$ and there is a budget $B>0$. Without loss of generality, we also assume $\forall u\in \N\colon c(u)\leq B$. Then a knapsack constraint can be modeled as an independence system as follows:
\begin{definition}[Knapsack Constraint]\label{definition:knapsack}
	An independence system $(\N,\I)$ capturing a knapsack constraint is defined as the collection of sets $X\subseteq \N$ obeying $c(X)\leq B$ for some non-negative modular function $c(X)=\sum_{u\in X}c(u)$. 
\end{definition}
Based on the above definitions, the SKP, SMC and SSP can be written as:
\begin{itemize}[leftmargin=8mm]
	\item SKP: $\max\{f(S)\colon S\in\I\}$, where $(\N,\I)$ is a knapsack constraint
	\item SMC: $\max\{f(S)\colon S\in\I\}$, where $(\N,\I)$ is a knapsack constraint and $\forall u\in\N:c(u)=1$	
	\item SSP: $\max\{f(S)\colon S\in\I\}$, where $(\N,\I)$ is a $k$-system	and $\forall u\in\N:c(u)=1$
\end{itemize}

For convenience, we denote by $\OPT=f(O)$ the optimal value of the objective function for the SKP/SSP/SMC, where $O$ is an optimal solution, and use $w$ to denote an element with maximum cost in $O$. Moreover, we let $r$ denote the maximum cardinality of any feasible solution to the same problem. For any $u\in \N$ and any $X\subseteq \N$, we use $f_{X}(\cdot)$ to denote the function defined as $f_X(Y)=f(X\cup Y)$ for any $Y\subseteq \N$; and we use $f(u\mid X)$ to denote the ``marginal gain'' of $u$ with respect to $X$, i.e., $f(u\mid X)=f(X\cup \{u\})-f(X)$; we also call $f(u\mid X)/c(u)$ as the ``marginal density'' of $u$ with respect to $X$.

Suppose that $f(S)$ can be returned by an oracle query for any given $S\subseteq \N$, the \textit{query complexity} of any algorithm $\mathit{ALG}$ denotes the number of oracle queries to $f(\cdot)$ incurred in $\mathit{ALG}$, and its \textit{adaptive complexity} denotes the number of \textit{adaptive rounds} of $\mathit{ALG}$, where $\OO(\mathrm{poly}(n))$ oracle queries are allowed in each adaptive round, but all these queries can only depend on the results of previous adaptive rounds. We assume that there exists an algorithm $\mathsf{USM}(X)$ addressing the \textit{unconstrained submodular maximization ({USM})} problem of $\max\{f(Y)\colon Y\subseteq X\}$ for any $X\subseteq \N$, and assume that it achieves the following performance bounds: 
\begin{theorem}[Theorem A.1 in the full version of \shortcite{chen2019unconstrained}]\label{thm:USM}
	For every constant $\epsilon>0$, there is an algorithm without using multi-linear extension that achieves a $(1/2-\epsilon)$-approximation for the unconstrained submodular maximization problem using $\OO(\frac{1}{\epsilon}\log\frac{1}{\epsilon})$ adaptive rounds with $\OO(\frac{n}{\epsilon^4}\log^3\frac{1}{\epsilon})$ query complexity.
\end{theorem}
\section{Approximation Algorithms}\label{sec:pt}


\bl{In this section, we propose our \algone~and \PS~algorithms. Both of them call a sub-module \RB~based on the ``adaptive sequencing'' technique, which was originally proposed by \shortcite{balkanski2019optimal}. It is worth noting that the \bl{``adaptive sequencing''} technique alone is far from sufficient to achieve approximation ratios for any submodular maximization problem, thus numerous studies such as~\shortcite{amanatidis2021submodular,amanatidis2023submodular,breuer2020fast,quinzan2021adaptive,gong2024algorithms} built upon this technique to develop low-adaptivity algorithms for submodular maximization under different constraints as we have done. In contrast to the original technique of \shortcite{balkanski2019optimal} and its variants such as that in~\shortcite{amanatidis2021submodular,amanatidis2023submodular}, our \RB~introduces a ``random batch selection'' method to allow for a more flexible adaptive sequencing. This enhancement expands the applicability of our algorithm to a wider range of constraints and objective functions. More importantly, our methods for generating candidate solutions (i.e., \algone~and \PS) differ from existing work, which allows us to avoid the issues encountered in \shortcite{amanatidis2021submodular,quinzan2021adaptive} and achieve better performance bounds than their algorithm.}

For clarity, we first introduce \RB~(Sec.~\ref{sec:rb}), then introduce \algone~(Sec.~\ref{sec:skp1}) and \PS~(Sec.~\ref{sec:ssp}), respectively.
\subsection{The RandBatch Procedure}\label{sec:rb}
\SetKwFor{With}{with}{do}{endw}
\begin{algorithm}[t]
	\caption{$\mathsf{RandBatch}(\rho, I, M, p, \epsilon, f(\cdot), c(\cdot))$}
	\label{alg:rb}
	\KwIn{density threshold $\rho$, set $I$, maximum number of iterations $M\in \mathbb{Z}_{>0}$, probability $p$, precision $\epsilon\in (0,1)$, submodular function $f(\cdot)$, and cost function $c(\cdot)$\;}
	$A\leftarrow \emptyset$; $U\leftarrow \emptyset$; $count\gets 0$\;
	$L\leftarrow\{u\in I\colon\frac{f(u\mid A)}{c(u)}\geq\rho\wedge A\cup\{u\}\in\I\}$\;\label{ln:initialize_L}
	\While{$L\neq\emptyset\wedge count<M$\label{ln:break}}
	{
		$\{v_1,v_2,\dotsc,v_{d}\}\leftarrow \mathsf{GetSEQ}(A,L,c(\cdot))$\;\label{ln:random1-1}
		\ForEach{$i\in\{0,1,\dotsc,d\}$\label{ln:find_start}}{
			$V_i\leftarrow\{v_1,v_2,\dotsc,v_i\}$; $G_i\gets A\cup V_i$\;\label{ln:def_V}
			$E_i^+\gets\{u\in L\colon\frac{f(u\mid G_i)}{c(u)}\geq\rho\wedge G_i\cup \{u\}\in \I\}$\;\label{ln:define_start}
			$E_i^-\gets\{u\in L\colon f(u\mid G_i)< 0\}$\;\label{ln:define_E-}
			$D_i\leftarrow\{v_{j}\colon j\in [i]\wedge f(v_j\mid A\cup V_{j-1})< 0\}$\;\label{ln:define_end}
		}
		Find $t_1\leftarrow \min_{i\leq d}\{c(E_i^+)\leq (1-\epsilon)c(L)\}$ and $t_2\leftarrow \min_{i\leq d}\{\epsilon \sum_{u\in E_i^+}f(u\mid G_i)\leq \sum_{u\in E_i^-}|f(u\mid G_i)|+\sum_{v_j\in D_i}|f(v_j\mid A\cup V_{j-1})|\}$\;\label{ln:search}
		$t^*\leftarrow\min\{t_1,t_2\}$; $U\leftarrow U\cup \{V_{t^*}\}$\;\label{ln:find_end}
		\With{probability $p$\label{ln:random2-1}}{
			$A\leftarrow A\cup V_{t^*}$\;\label{ln:addbiga}
			\lIf{$t_2<t_1$}{$count \leftarrow count+1$\label{ln:count}}
		}
		$L\gets \{u\colon u\in L\setminus U\wedge \frac{f(u\mid A)}{c(u)}\geq \rho\wedge A\cup\{u\}\in\I \}$\label{ln:update_L}\;
	}
	\Return{$(A,U,L)$}\;
\end{algorithm}
\begin{algorithm}[t]
	\caption{$\mathsf{GetSEQ}(A,I)$}
	\label{alg:gs}
 	\KwIn{sets $A$ and $I$\;}
	$V\leftarrow\emptyset$\;
	\While{$I\neq\emptyset$}{
		Randomly permute $I$ into $\{v_1,\dotsc, v_{|I|}\}$\;
		$s\leftarrow \max\{i\in [|I|]:A\cup V\cup \{v_1,\dotsc,v_i\}\in \I\}$\;
		$V\leftarrow V\cup \{v_1,\dotsc,v_s\}$\;
		$I\leftarrow\{u\colon u\in I\setminus V\wedge A\cup V\cup\{u\}\in\I\}$\;
	}
	\Return{$V$}\;	
\end{algorithm}



The \RB~procedure (Algorithm~\ref{alg:rb}) takes as input a threshold $\rho$, candidate element set $I$, submodular function $f(\cdot)$, cost function $c(\cdot)$, a number $M$ to control the adaptivity, and $p,\epsilon\in (0,1]$. 
It runs in iterations to find a solution set $A$ and uses $U$ to store all elements considered for addition to $A$. Besides, it maintains a set $L$ of ``valuable elements'' in $I$ that have not been considered throughout the procedure, where any element is called a \textit{valuable element} w.r.t.\ $A$ if it can be added into $A$ with marginal density no less than $\rho$ under the budget constraint (Line~\ref{ln:initialize_L}). In each iteration, \RB~first neglects $f(\cdot)$ and calls a simple function \RQ~(Algorithm~\ref{alg:gs}) to get a random sequence of elements $(v_1,\dotsc, v_d)$ from $L$ without violating the constraint (Line~\ref{ln:random1-1}), and then finds a subsequence $V_{t^*}=(v_1,\dotsc, v_{t^*})$ with ``good quality'' by considering $f(\cdot)$ (Lines~\ref{ln:find_start}--\ref{ln:find_end}), where $t^*=\min\{t_1,t_2\}$ will be explained shortly. After that, it invokes a ``random batch selection'' operation by adding $V_{t^*}$ into $A$ with probability of $p$ and abandoning $V_{t^*}$ with probability of $1-p$ (Line~\ref{ln:random2-1}). All the elements in $V_{t^*}$ are recorded into $U$ no matter they are accepted or abandoned. Then \RB~enters a new iteration and repeats the above process with an updated $L$ (Line~\ref{ln:update_L}). \RB~uses a variable $count$ to control its adaptive complexity (Line~\ref{ln:count}), and returns $(A, U, L)$ either when $L=\emptyset$ or $count=M$. Note that \RB~returns $L\neq \emptyset$ only when $count=M$.

As mentioned above, \RB~uses $t^*=\min\{t_1,t_2\}$ to control the quality of the elements in $V_{t^*}$, where $t_1, t_2$ depend on $E_i^+$ (elements in $L$ with enough density), $E_i^-$ (elements in $L$ with negative marginal gain) and $D_i$ (elements in $V_i$ with negative marginal gain) defined in Lines~\ref{ln:define_start}-\ref{ln:define_end}. Intuitively, the setting of $t_1$ (Line~\ref{ln:search}) ensures that the total cost of valuable elements w.r.t. $A\cup V_{t_1}$ (i.e., elements in $E_{t_1}^+$) is sufficiently small, and the setting of $t_2$ (Line~\ref{ln:search}) ensures that the total marginal gain of valuable elements w.r.t. $A\cup V_{t_1}$ (i.e., elements in $E_{t_2}^+$) is sufficiently small. Through the selection of $V_{t^*}$, \RB~strikes a balance between solution quality and adaptive complexity, as shown by the following lemma:
\begin{lemma}\label{lm:properties}
	The sets $A$ and $L$ output by $\mathsf{RandBatch}(\rho,$ $I, M,f(\cdot),c(\cdot), p, \epsilon)$ satisfy (1) $A\in\I$, (2) $\mathbb{E}[f(A)]\geq (1-\epsilon)^2\rho\cdot \mathbb{E}[c(A)]$ and (3) $\epsilon\cdot M\cdot\sum_{u\in L}f(u\mid A)\leq \OPT$ for any $I\subseteq \N$.
\end{lemma}
It can be readily observed that property 1 of Lemma \ref{lm:properties} holds as \RQ~always returns feasible sequences. However, the proof of properties 2 and 3 of Lemma \ref{lm:properties} are a bit complex due to the randomness introduced in \RQ~function and in Line~\ref{ln:random2-1} of \RB. To overcome this challenge, we first introduce Lemma \ref{lm:density}, which shows that the expected marginal density of each element that \RB~attempts to add to the candidate solution is sufficiently large. We then use Lemma \ref{lm:density} to prove Lemma \ref{lm:properties}. The proof is inspired by \shortcite{balkanski2019optimal,amanatidis2021submodular}, but is more involved due to the randomness introduced by Line \ref{ln:random2-1}
of Algorithm \ref{alg:rb}.
\begin{lemma}\label{lm:density}
	For any $V_{t^*}$ found in Lines~\ref{ln:search}--\ref{ln:find_end} of \RB~and any $u\in V_{t^*}$, let $\lambda(u)$ denote the set of elements in $V_{t^*}$ selected before $u$ (note that $V_{t^*}$ is an ordered list according to Line~\ref{ln:def_V}), and $\lambda(u)$ does not include $u$. Let $A$ and $U$ denote the sets returned by $\mathsf{RandBatch}(\rho, I, M, p, \epsilon, f(\cdot), c(\cdot))$, where the elements in $U$ are $\{u_1,u_2,\dotsc, u_s\}$ (listed according to the order they are added into $U$). Let $u_j$ be a ``dummy element'' with zero marginal gain and zero cost for all $s<j\leq |I|$. Then we have:
	\begin{eqnarray}
	\nonumber	\forall j\in [|I|]\colon &&\E [f(u_j\mid \{u_1\dotsc, u_{j-1}\}\cap A\cup \lambda(u_j))\mid \mathcal{F}_{j-1}]\geq (1-\epsilon)^2\rho \cdot \E[c(u_j)\mid \mathcal{F}_{j-1}],
	\end{eqnarray}
	where $\mathcal{F}_{j-1}$ denotes the filtration capturing all the random choices made until the moment right before selecting $u_j$.
\end{lemma}
\begin{proof}[Proof of Lemma~\ref{lm:density}]
	Note that $\mathcal{F}_{j-1}$ determines whether $u_j\in U$, so the lemma trivially follows for all $j>s$. In the sequel, we consider the case of $j\leq s$. Recall that \RB~runs in iterations and adds a batch of elements $V_{t^*}$ into $U$ in each iteration (Line~\ref{ln:find_end}). Consider the specific iteration in which $u_j$ is added into $U$ and the sets $G_q, E_{q}^+,E_{q}^-$ defined by Lines~\ref{ln:def_V}--\ref{ln:define_E-} of \RB~in that iteration, where $q=|\lambda(u_j)|$; and let $H=\{v\colon v\in L\setminus G_{q}\wedge G_{q}\cup \{v\}\in \I\}$, where $L$ is the set considered at the beginning of that iteration. So we have $\{u_1,\dotsc, u_{j-1}\}\cap A\cup \lambda(u_j)=G_q$. Note that both $G_{q}$ and $H$ are deterministic given $\mathcal{F}_{j-1}$, and that $u_j$ is drawn uniformly at random from $H$. So we have
	\begin{align}
	\nonumber	& \E\big[{f(u_j\mid \{u_1,\dotsc, u_{j-1}\}\cap A\cup \lambda(u_j))\mid \mathcal{F}_{j-1}}\big]-(1-\epsilon)^2\rho\cdot\E\big[c(u_j)\mid \mathcal{F}_{j-1}\big]\\
	\nonumber &= \sum\nolimits_{v\in H}\Pr[u_j=v\mid \mathcal{F}_{j-1}] f(v\mid G_q)-(1-\epsilon)^2\rho\sum\nolimits_{v\in H}\Pr[u_j=v\mid \mathcal{F}_{j-1}] c(v)\\
	\nonumber &\geq |H|^{-1}\cdot\bigg(\sum\nolimits_{v\in E_{q}^+} f(v\mid G_{q})+\sum\nolimits_{v\in E_{q}^-} f(v\mid G_{q})-(1-\epsilon)^2\rho\cdot\sum\nolimits_{v\in H}c(v)\bigg)\\
	&=|H|^{-1}\cdot\left(\epsilon\sum\nolimits_{v\in E_{q}^+}f(v\mid G_{q})-\sum\nolimits_{v\in E_{q}^-}|f(v\mid G_{q})|\right)\label{eqn:non-nega-1}\\
	&\mathrel{\phantom{=}}\mathop{+}|H|^{-1}\cdot (1-\epsilon)\cdot \sum\nolimits_{v\in E_{q}^+}\left(f(v\mid G_{q})-\rho\cdot c(v)\right)\label{eqn:non-nega-2}\\
	&\mathrel{\phantom{=}}\mathop{+}|H|^{-1}\cdot (1-\epsilon)\cdot \rho\cdot \left(\sum\nolimits_{v\in E_{q}^+}c(v)-(1-\epsilon)\sum\nolimits_{v\in H} c(v)\right),\label{eqn:non-nega-3}
	\end{align}
	where the first inequality is due to $E_{q}^+\cup E_{q}^-\subseteq H$. Note that Eqn.~\eqref{eqn:non-nega-2} is non-negative due to the definition of $E_{q}^+$. According to the definition of $t^*$ in Lines~\ref{ln:search}--\ref{ln:find_end} of Algorithm~\ref{alg:rb}, we have $c(E_{q}^+)>(1-\epsilon)c(L)$ due to $q\leq t^*-1$, so Eqn.~\eqref{eqn:non-nega-3} is non-negative as $H\subseteq L$. Similarly, Eqn.~\eqref{eqn:non-nega-1} is also non-negative due to the definition of $t^*$. Combining these completes the proof.
\end{proof}
\begin{proof} [Proof of Lemma~\ref{lm:properties}]
	We first prove $\mathbb{E}[f(A)]\geq (1-\epsilon)^2\rho\cdot \mathbb{E}[c(A)]$. Consider the sequence $\{u_1,\dotsc, u_s, u_{s+1},\dotsc, u_{|I|}\}$ defined in Lemma~\ref{lm:density}. For each $j\in [|I|]$, define a random variable $\delta_1(u_j)= f(u_j\mid \{u_1,\dotsc, u_{j-1}\}\cap A\cup \lambda(u_j))$ if $u_j\in A$ and $\delta_1(u_j)=0$ otherwise, and also define $\delta_2(u_j)= c(u_j)$ if $u_j\in A$ and $\delta_2(u_j)=0$ otherwise. So we have $f(A)\geq\sum\nolimits_{j=1}^{|I|} \delta_1(u_j)$ and $c(A)=\sum\nolimits_{j=1}^{|I|}\delta_2(u_j)$. Due to the linearity of expectation and the law of total expectation, we only need to prove
	\begin{equation}
	\forall j\in [|I|], \forall \mathcal{F}_{j-1}\colon \E [\delta_1(u_j)\mid \mathcal{F}_{j-1}] \geq (1-\epsilon)^2\rho\cdot  \E [\delta_2(u_j)\mid \mathcal{F}_{j-1}],\label{eqn:involvedeqn}
	\end{equation}
	where $\mathcal{F}_{j-1}$ is the filtration defined in Lemma~\ref{lm:density}. This trivially holds for $j>s$. For any $j\leq s$, we have
	\begin{eqnarray}
	\nonumber\E [\delta_1(u_j)| \mathcal{F}_{j-1}]&=&\E [\E[\delta_1(u_j)| \mathcal{F}_{j-1}, u_j]\mid \mathcal{F}_{j-1}]\\
	\nonumber&=&p\E [f(u_j\mid \{u_1,\dotsc, u_{j-1}\}\cap A\cup \lambda(u_j))\mid \mathcal{F}_{j-1}],
	\end{eqnarray}
	where the second equality is due to the reason that, each $u_j\in U$ is added into $A$ with probability of $p$, which is independent of the selection of $u_j$. Similarly, we can prove $\E [\delta_2(u_j)\mid \mathcal{F}_{j-1}]=p\cdot \E [c(u_j)\mid \mathcal{F}_{j-1}]$. Combining these results with Lemma~\ref{lm:density} proves Eqn.~\eqref{eqn:involvedeqn} and hence $\mathbb{E}[f(A)]\geq (1-\epsilon)^2\rho\cdot \mathbb{E}[c(A)]$.
	
	Next, we prove $\epsilon\cdot M\cdot\sum_{u\in L}f(u\mid A)\leq f(O)$. This trivially holds if $L=\emptyset$, otherwise we must have $count=M$ when Algorithm~\ref{alg:rb} returns $(A,U,L)$ due to Line~\ref{ln:break}. Note that $count$ is increased by 1 only in Line~\ref{ln:count} of the while-loop in Algorithm~\ref{alg:rb}. Consider the $i$-th iteration among the specific $M$ iterations of the while-loop in which $count$ gets increased,  
	and let $E_{[i]}^+, E_{[i]}^-, D_{[i]}$ denote the sets $E_{t^*}^+, E_{t^*}^-, D_{t^*}$ in that iteration, and let $A_{[i]}$ denote the set of elements already added into $A$ at the end of that iteration. Besides, let $A^<(u)$ denote the elements in $A$ that are selected before $u$ for any $u\in A$, and let $A^+=\{u\in A\colon f(u\mid A^<(u))\geq 0\}$ and $A^-=\{u\in A\colon f(u\mid A^<(u))< 0\}$. As the sets in $\{E_{[i]}^-, D_{[i]}\colon i\in [M]\}$ are mutually disjoint according to their definitions and $\bigcup\nolimits_{i=1}^M E_{[i]}^-\cap A=\emptyset$, we can use submodularity of $f(\cdot)$ to get
	\begin{align}
	&f(\bigcup\nolimits_{i=1}^M E_{[i]}^-\cup A)\nonumber\\
	&\leq f(A)+\sum\nolimits_{i=1}^M\sum\nolimits_{u\in E_{[i]}^-}f(u\mid A) \nonumber\\
	\nonumber&\leq\sum\nolimits_{u\in A^+}f(u\mid A^<(u))+\sum\nolimits_{u\in A^-}f(u\mid A^<(u))+\sum\nolimits_{i=1}^M\sum\nolimits_{u\in E_{[i]}^-}f(u\mid A_{[i]})\\
	\nonumber&\leq\sum\nolimits_{u\in A^+}f(u\mid A^<(u))+\sum\nolimits_{i=1}^M\sum\nolimits_{u\in D_{[i]}}f(u\mid A^<(u))\\
	&+\sum\nolimits_{i=1}^M\sum\nolimits_{u\in E_{[i]}^-}f(u\mid A_{[i]}),~~\label{eqn:bound_E}
	\end{align}
	where the last inequality follows from the fact that $\cup_{i\in [M]}D_{[i]}\subseteq A^-$. Combining Eqn.~\eqref{eqn:bound_E} with $\sum_{u\in A^+}f(u\mid A^<(u))\leq f(A^+)$ (due to submodularity) and $f(\bigcup\nolimits_{i=1}^M E_{[i]}^-\cup A)\geq 0$, we can get
	\begin{equation}
	\sum_{i=1}^M\Big(\sum\limits_{u\in D_{[i]}}|f(u\mid A^<(u))|+\sum_{u\in E_{[i]}^-}|f(u\mid A_{[i]})|\Big)\leq f(A^+)\leq \OPT.\label{eqn:bound_A+}
	\end{equation}
	Besides, according to Lines~\ref{ln:search},\ref{ln:find_end},\ref{ln:count} of Algorithm~\ref{alg:rb}, we must have
	\begin{equation}
	\forall i\in [M]\colon \epsilon\sum_{u\in E_{[i]}^+} f(u\mid A_{[i]})\leq \sum_{u\in E_{[i]}^-}|f(u\mid A_{[i]})|+\sum_{u\in D_{[i]}}|f(u\mid A^<(u))|.\label{eqn:value_inequ}
	\end{equation}
	Moreover, we have $\sum_{i=1}^M\sum\nolimits_{u\in E_{[i]}^+} f(u\mid A_{[i]})\geq M\cdot \sum_{u\in E_{[M]}^+}f(u\mid A)$ due to submodularity of $f(\cdot)$ and $E_{[M]}^+\subseteq E_{[M-1]}^+\subseteq\dotsb\subseteq E_{[1]}^+$. Combining this with Eqn.~\eqref{eqn:bound_A+} and Eqn.~\eqref{eqn:value_inequ} yields $\OPT\geq \epsilon\cdot M\cdot \sum_{u\in E_{[M]}^+}f(u\mid A)$, which completes the proof due to $E_{[M]}^+=L$ when Algorithm~\ref{alg:rb} returns a non-empty $L$.
\end{proof}
The complexity of \RB~(shown in Lemma~\ref{lma:complexityofrandbatch}) can be proved by using the fact that, when $A$ enlarges, either $c(L)$ is decreased by a $1-\epsilon$ factor, or $count$ is increased by 1 (Line~\ref{ln:count}).
\begin{lemma} \label{lma:complexityofrandbatch}
	\RB~has $\OO((\frac{1}{\epsilon}\log(|I|\cdot \beta(I))+M)/ p)$ adaptivity, and its query complexity is $\OO(|I|\cdot r)$ times of its adaptive complexity, where $\beta(I)\triangleq\max_{u,v\in I}\frac{c(u)}{c(v)}$. If we use binary search in Line \ref{ln:search}, then \RB~has $\OO((\frac{1}{\epsilon}$
	$\log(|I|\cdot \beta(I))+M)\cdot(\log r)/ p)$ adaptivity, and its query complexity is $\OO(|I|)$ times of its adaptivity.
\end{lemma}
\begin{proof}
	First, we analyze the number of while-loops in \RB. Note that \RB~needs $\OO(1/p)$ while-loops (in expectation) to trigger Lines~\ref{ln:addbiga}--\ref{ln:count} once. Each time when Lines~\ref{ln:addbiga}--\ref{ln:count} are executed, either $c(L)$ is decreased by at least a $1-\epsilon$ factor, or $count$ is increased by $1$. Besides, note that \RB~terminates either when $L=\emptyset$ or $count=M$. At the beginning of the algorithm, we have $c(L)\leq c_{max}\cdot |I|$ where $c_{max}=\max_{u\in I}c(u)$, while we need $c(L)<\min_{u\in I}c(u)$ to ensure $L=\emptyset$. Based on the above discussions, it can be seen that the total number of while-loops in \RB~is at most $\OO((\frac{1}{\epsilon}\log(|I|\cdot \beta(I))+M)/p)$ in expectation.
	
	Second, we explain why binary search can be used in Line \ref{ln:search} of Algorithm \ref{alg:rb}. Recall that Line \ref{ln:search} need to find the smallest $i\in [d]$ satisfying Eqn.~\eqref{eqn:cost_condition} to determine $t_1$, and to find the smallest $i\in [d]$ satisfying Eqn.~\eqref{eqn:value_condition} to determine $t_2$:
	\begin{align}
	&c(E_i^+)\leq (1-\epsilon)c(L),\label{eqn:cost_condition}\\
	&\epsilon \sum\nolimits_{u\in E_i^+}f(u\mid G_i)\leq \sum\nolimits_{u\in E_i^-}|f(u\mid G_i)|+\sum\nolimits_{v_j\in D_i}|f(v_j\mid A\cup V_{j-1})|.\label{eqn:value_condition}
	\end{align}
	By submodularity and the definitions of $E_i^+,E_i^-,D_i$, it can be verified that the LHS of Eqn.~\eqref{eqn:cost_condition} or Eqn.~\eqref{eqn:value_condition} decreases when $i$ increases, and the RHS of Eqn.~\eqref{eqn:value_condition} increases with $i$. This makes it possible to use binary search.
	
	Third, we analyze the adaptive complexity and query complexity caused by each while-loop in \RB. Note that the $\mathsf{GetSEQ}$ function causes zero adaptive complexity. Therefore, if binary search is not used in Line \ref{ln:search}, then the seeking of $t^*$ in Lines~\ref{ln:search}--\ref{ln:find_end} of \RB~can be fully parallelized and hence causes $\OO(1)$ adaptive complexity, but under $\OO(|I|\cdot r)$ query complexity because $\OO(|I|)$ oracle queries are needed to calculate $E_i^+, E_i^-$ for each $i\leq d\leq r$. On the other side, if binary search is used, then we need $\OO(\log r)$ adaptive rounds to find $t^*$ due to $t^*\leq d\leq r$, causing $\OO(|I|\log r)$ oracle queries. The lemma then follows by synthesizing all the above discussions.
\end{proof}
\subsection{The Algorithm for the SKP}\label{sec:skp1}
\begin{algorithm}[t]
	\caption{$\mathsf{ParSKP}(\alpha,\epsilon,B,f(\cdot),c(\cdot))$}
	\label{alg:skp1}
	\KwIn{parameter $\alpha\in(0,1)$, precision $\epsilon\in (0,1)$, budget $B$, submodular function $f(\cdot)$, and cost function $c(\cdot)$\;}
 $\N_{1}\leftarrow\{u\in\N\colon c(u)>\epsilon\frac{B}{n}\}$; $\N_{2}\leftarrow\N\setminus \N_1$\;
	$u^*\leftarrow\arg\max_{u\in\N}{f(u)}$; $S\gets\mathsf{USM}(\mathcal{N}_2)$\;\label{ln:usm_in_skp1}
	$S\gets\arg\max_{X\in \{S,\{u^*\}\}}f(X)$\;\label{ln:u_star_skp1}
	$\rho_{min}\leftarrow{\frac{\alpha  f(u^*)}{B}};\rho_{max}\leftarrow{\frac{n^2\cdot\alpha  f(u^*)}{\epsilon B}}$\;\label{ln:rho}
	$Z\leftarrow\{(1-\epsilon)^{-z}\colon z\in \mathbb{Z}\wedge (1-\epsilon)^{-z}\in[\rho_{min},\rho_{max}]\}$\;\label{ln:Z}
	\ForEach{$\rho\in Z$ in parallel\label{ln:for_start_skp1}}{
		\For{$i\leftarrow 1$ {\rm\bfseries to} $\lceil \log_{1-\epsilon}\epsilon\rceil$ in parallel\label{ln:multirun_start}}{
			$T\gets \mathsf{Probe}(\rho,\mathcal{N}_1, \mathcal{N}_2, \epsilon,B,f(\cdot),c(\cdot))$\;
			$S\gets \arg\max_{X\in \{S,T\}}f(X)$\;
		}\label{ln:multirun_end}
	}\label{ln:for_end_skp1}
	\Return{$S$}\;
\end{algorithm}
\begin{algorithm}[t]
	\caption{$\mathsf{Probe}(\rho,\N_1,\N_2, \epsilon,B,f(\cdot),c(\cdot))$}
	\label{alg:probe}
 	\KwIn{density threshold $\rho$, two disjoint partitions of the ground set $\N_1$ and $\N_2$, precision $\epsilon\in (0,1)$, budget $B$, submodular function $f(\cdot)$, and cost function $c(\cdot)$\;}
	$T\gets \emptyset$; $M\gets \lceil \epsilon^{-2}\rceil$; $p\gets 1$; $I\gets \N_1$\;
	$(A_1, U_1, L_1)\gets\mathsf{RandBatch}(\rho,I,M,p,\epsilon,f(\cdot),c(\cdot))$\;\label{ln:A_1}
	$I\gets \N_1\setminus A_1$\;
	$(A_2, U_2, L_2)\gets \mathsf{RandBatch}(\rho, I,M,p,\epsilon,f(\cdot),c(\cdot))$\; \label{ln:A_2}
	\For{$i\leftarrow 1$ {\rm\bfseries to} $2$\label{ln:augment_start}}{
		$e_i\leftarrow\arg\max_{u\in\N_1\wedge c(A_i\cup\{u\})\leq B}f(A_i\cup\{u\})$\;\label{ln:augment}
		$T\leftarrow\arg\max_{X\in\{T, A_i,A_i\cup\{e_i\}\}}f(X)$\;
	}\label{ln:augment_end}
	\If{$c(\N_2\cup A_1)\leq B$}{
		$A_3\leftarrow \mathsf{USM}(\N_2\cup A_1)$\; \label{ln:usm_in_probe}
		{$T\leftarrow\arg\max_{X\in\{T, A_3\}}f(X)$}\;
	}
	\Return{$T$}\;
\end{algorithm}
Our \algone~algorithm is shown in Algorithm~\ref{alg:skp1}. In \algone, the ground set $\N$ is partitioned into two disjoint subsets $\N_1$ and $\N_2$, where $\N_1$ contains every element in $\N$ with a sufficiently large cost (i.e., larger than $\epsilon\cdot B/n$). So we have $c(\N_2)\leq \epsilon\cdot B$. A major building block of \algone~is the function \onerun~(shown in Algorithm~\ref{alg:probe}). For clarity, we first elaborate \onerun~in the following.

Given an input threshold $\rho$ and $\N_1, \N_2$, \onerun~first calls \RB~with $p=1$ using $\N_1$ as the ground set to find a candidate solution $A_1$ (Line~\ref{ln:A_1}), and then calls \RB~again using $\N_1\setminus A_1$ as the ground set to find another candidate solution $A_2$ (Line~\ref{ln:A_2}). So $A_1$ and $A_2$ are disjoint subsets of $\N_1$. The reason for calling \RB~with only the elements in $\N_1$ is that the adaptive complexity of \RB~can be bounded only when the costs of considered elements have a lower bound (due to Lemma~\ref{lma:complexityofrandbatch}). Then, \onerun~tries to ``boost'' the utilities of $A_1$ and $A_2$ by augmenting them with a single element in $\N_1$, neglecting the threshold $\rho$ (Line~\ref{ln:augment}). After that, another candidate solution set $A_3$ is found by calling an unconstrained submodular maximization algorithm if $c(\N_2\cup A_1)\leq B$ (Line~\ref{ln:usm_in_probe}). Finally, \onerun~returns the candidate solution with maximum function value found so far.

In Lemma~\ref{lm:onerun}, we show \onerun~can achieve a provable approximation ratio under some special cases:

\begin{lemma}\label{lm:onerun}
	If the threshold $\rho$ input into Algorithm~\ref{alg:probe} is no more than $\rho^*\triangleq\frac{\alpha f(O)}{B-c(w)}$ and Algorithm~\ref{alg:probe} finds $A_1$ and $A_2$ satisfying $\forall i\in \{1,2\}\colon c(A_i)< B-\max\{\epsilon B, c(w)\}$, where $w$ is the element in $O$ with the maximum cost, then Algorithm~\ref{alg:probe} returns a solution $T$ satisfying $f(T)\geq \frac{(1-2\epsilon)(1-2\alpha-2\epsilon)}{4-4\epsilon}\cdot\OPT$.
\end{lemma}
\begin{proof}
	According to the assumption of the lemma, we must have $c(A_1\cup \N_2)\leq B$ due to $c(\N_2)\leq \epsilon\cdot B$, so Line~\ref{ln:usm_in_probe} of \onerun~must be executed. If $w\notin \N_1$, then we must have $O\subseteq \N_2$ and hence $2f(T)\geq 2f(A_3)\geq (1-2\epsilon)f(O)$ due to Line~\ref{ln:usm_in_probe} of Algorithm~\ref{alg:probe} (Theorem~\ref{thm:USM}), which completes the proof. Therefore, we assume $w\in \N_1$ in the following.
	
	Note that $\N_1\cap \N_2=\emptyset, A_1\cap A_2=\emptyset$ and $A_1,A_2\subseteq \N_1$. So we can use submodularity of $f(\cdot)$ to get:
	\begin{align}
	\nonumber f(O)&\leq f(O\cap\N_{1}\setminus A_1)+f((O\cap A_1) \cup (O\cap \N_2))\\
	&\leq f(A_2\cup(O\cap\N_{1}\setminus A_1))
	+f(A_1\cup(O\cap\N_{1})) +f((O\cap A_1) \cup (O\cap \N_2)).\label{eqn:bound_O_A_1}
	\end{align}
	Next, we try to bound the three additive factors in the RHS of Eqn.~\eqref{eqn:bound_O_A_1}. For the third additive factor, we have
	\begin{equation}
	(1-2\epsilon)f((O\cap A_1) \cup (O\cap \N_2))\leq 2f(A_3)\leq 2f(T),\label{eqn:Q} 
	\end{equation}
	due to Line~\ref{ln:usm_in_probe} of Algorithm~\ref{alg:probe}. Besides, we can get
	\begin{align}
	\nonumber&f(A_1\cup(O\cap\N_{1}))\leq f(A_1\cup \{w\})+ \sum\nolimits_{u\in Q} f(u\mid A_1\cup \{w\})\leq f(T)+ \sum\nolimits_{u\in Q} f(u\mid A_1) \\
	&\leq f(T)+ \sum\nolimits_{u\in L_1} f(u\mid A_1)+\sum\nolimits_{u\in Q\setminus L_1} f(u\mid A_1),\label{eqn:bound_A1}
	\end{align}
	where $Q=O\cap \N_1\setminus (A_1\cup \{w\})$, and $L_1$ is the set returned by \RB~in Line~\ref{ln:A_1} of \onerun, and the second inequality is due to the the submodularity of $f(\cdot)$ and Lines~\ref{ln:augment_start}--\ref{ln:augment_end} of Algorithm~\ref{alg:probe}. Furthermore, note that each $u\in Q\setminus L_1$ satisfies $c(A_1\cup \{u\})\leq B$ according to the assumption of current lemma, so we should have $\frac{f(u\mid A_1)}{c(u)}<\rho$, because otherwise $u$ should be in either $A_1$ or $L_1$ due to the design of \RB. Using this and $c(Q)\leq B-c(w)$, we get
	\begin{equation}
	\sum\nolimits_{u\in Q\setminus L_1} f(u\mid A_1)\leq \rho\cdot c(Q)\leq \rho^*\cdot c(Q)\leq \alpha\cdot f(O).\label{eqn:bound_L1}
	\end{equation}
	Besides, we can use Lemma~\ref{lm:properties} to get $\sum_{u\in L_1} f(u\mid A_1)\leq \epsilon\cdot f(O)$ due to $M=\lceil \epsilon^{-2}\rceil$. Combining this with Eqn.~\eqref{eqn:bound_A1} and Eqn.~\eqref{eqn:bound_L1} yields
	\begin{equation}
	f(A_1\cup(O\cap\N_{1}))\leq f(T)+ (\alpha+\epsilon)\cdot f(O).\label{eqn:A_1}
	\end{equation}
	Using similar reasoning as above, we can also get
	\begin{equation}
	f(A_2\cup(O\cap\N_{1}\setminus A_1))\leq f(T)+ (\alpha+\epsilon)\cdot f(O).\label{eqn:A_2}
	\end{equation}
	The lemma then follows by combining Eqns.~\eqref{eqn:bound_O_A_1}--\eqref{eqn:A_2}.
\end{proof}
There are still two obstacles for using Lemma~\ref{lm:onerun} to find the approximation ratio of \algone: the first problem is that $\rho^*$ is unknown, and the second problem is that $c(A_1)$ and $c(A_2)$ may not satisfy the condition in Lemma~\ref{lm:onerun}. In the following, we roughly explain how \algone~is designed to overcome these hurdles.

For the first problem mentioned above, it can be proved that $\rho^*\in [\rho_{min},$
$ \rho_{max}]$ if $B-c(w)> \epsilon B/n$, where $\rho_{min}$ and $\rho_{max}$ are defined in Line~\ref{ln:rho} of \algone. Therefore, \algone~tests multiple values of $\rho$ in $Z$ (Line~\ref{ln:Z}) to ensure that one of them lies in $[(1-\epsilon)\rho^*,\rho^*]$.
One the other side, if $B-c(w)\leq \epsilon B/n$, then we have $O\setminus \{w\}\subseteq \N_2$ and hence $\mathsf{USM}(\N_2)$ in Line~\ref{ln:usm_in_skp1} of \algone~can be used to find a ratio.
%
%
To address the second problem mentioned above, \algone~repeatedly runs \onerun~for a sufficiently large number of times (Lines~\ref{ln:multirun_start}--\ref{ln:multirun_end}). Therefore, if both $\E[c(A_1)]$ and $\E[c(A_2)]$ are sufficiently small, then it can be proved that at least one run of \onerun~satisfies the condition in Lemma~\ref{lm:onerun} with high probability. On the other side, if either $\E[c(A_1)]$ or $\E[c(A_2)]$ is sufficiently large, then we can directly use Lemma~\ref{lm:properties} to prove that \onerun~also satisfies a desired approximation ratio (in expectation). By choosing an appropriate $\alpha$ that can maximize our approximation ratio and combining all above ideas, we get:

\begin{theorem} \label{thm:boundofalgone}
	For the non-monotone SKP, \algone~can return a solution $S$ satisfying $\mathbb{E}[f(S)]\geq (\frac{1}{8}-\epsilon)\OPT$ by setting $\alpha=\frac{1}{4}$.
\end{theorem}
%
%
\begin{proof}
	If $c(O\setminus\{w\})\leq \epsilon\frac{B}{n}$, then it follows that $O\setminus\{w\}\subseteq \N_2$. As a result, Line~\ref{ln:usm_in_skp1} of Algorithm~\ref{alg:skp1} guarantees $f(S)\geq (1/2-\epsilon)f(O\setminus\{w\})$, which implies
	\begin{equation*}
	f(O)\leq f(w)+f(O\setminus\{w\})\leq f(\{u^*\})+\frac{2}{1-2\epsilon}f(S)\leq \bigg(1+\frac{2}{1-2\epsilon}\bigg)f(S)
	\end{equation*}
	where the second inequality is due to Line~\ref{ln:u_star_skp1} of Algorithm~\ref{alg:skp1}. Thus, we have $f(S)\geq \big(1+\frac{2}{1-2\epsilon}\big)^{-1}f(O)\geq (1/3-\epsilon)f(O)$ when $\epsilon<5/6$, which implies the theorem holds in this case.
	
	Then we only need to consider the case that $c(O\setminus\{w\})> \epsilon\frac{B}{n}$ in the follwing analysis. In this case, we have $B-c(w)\geq c(O)-c(w)> \epsilon\frac{B}{n}$ and hence
	\begin{equation*}
	\rho_{min}\!=\!\frac{\alpha f(u^*)}{B}\!\leq\!\frac{\alpha f(O)}{B-c(w)}\!=\!\rho^*\!\leq \!\frac{n\cdot\alpha f(u^*)}{B-c(w)}\!\leq\!\frac{n\cdot\alpha f(u^*)}{\epsilon\frac{B}{n}}\!\leq\! \frac{n^2\cdot\alpha f(u^*)}{\epsilon B}\!=\!\rho_{max}
	\end{equation*}
	which implies that there exists a threshold $\rho\in Z$ for the \onerun~algorithm such that $(1-\epsilon)\rho^*\leq\rho\leq \rho^*$, as per the definition of $Z$. With this threshold established, we can proceed with the following discussion:
	\begin{itemize}
		\item Case 1: The \onerun~algorithms finds $\{A_1,A_2\}$ such that there exists $A\in\{A_1,A_2\}$ satisfying $\mathbb{E}\big[c(A)\big]\geq (1-\epsilon)(B-\max\{\epsilon B, c(w)\})/2$. Under this case, we can use Lemma~\ref{lm:properties} to get
		\begin{align}
		\nonumber&\E\big[f(S)\big]\geq \E\big[f(A)\big]\geq(1-\epsilon)^2\rho\cdot \mathbb{E}\big[c(A)\big]\geq {(1-\epsilon)}^4\frac{\alpha f(O)}{B-c(w)}\cdot \frac{B-\max\{\epsilon B, c(w)\}}{2}\nonumber\\
		&={(1-\epsilon)}^4\frac{\alpha f(O)}{2}\cdot \frac{B-\max\{\epsilon B, c(w)\}}{B-c(w)}.
		\end{align}
		When $\epsilon B\geq c(w)$, we have
		\begin{equation}
		\frac{B-\max\{\epsilon B, c(w)\}}{B-c(w)}\!=\! \frac{B-\epsilon B}{B-c(w)}\!\geq\! 1-\frac{\epsilon B}{B-c(w)}\!\geq\! 1-\frac{\epsilon B}{B-\epsilon B}\!=\!\frac{1-2\epsilon}{1-\epsilon},
		\end{equation}
		and hence $\E[f(S)]\geq {(1-\epsilon)}^3(1-2\epsilon)\frac{\alpha}{2} f(O)\geq (1/8-\epsilon)f(O)$ when $\alpha= 1/4$. Similarly, we can prove $\E[f(S)]\geq (1/8-\epsilon)f(O)$ when $\epsilon B< c(w)$.
		\item Case 2: The \onerun~algorithms finds $A_1$ and $A_2$ satisfying $\mathbb{E}\big[c(A_i)\big]< (1-\epsilon)(B-\max\{\epsilon B, c(w)\})/2$ for all $i\in \{1,2\}$. Under this case, define an event $\mathcal{E}=\{\max\{c(A_1),$
		$c(A_2)\}\geq B-\max\{\epsilon B, c(w)\}\}$, then we have
		\begin{align}
		\nonumber	&(1-\epsilon)(B-\max\{\epsilon B, c(w)\})>\E\big[c(A_1)+c(A_2)\big]\\
		&\geq \E[c(A_1)+c(A_2)\mid \mathcal{E}]\cdot \Pr[\mathcal{E}]\geq (B-\max\{\epsilon B, c(w)\})\cdot\Pr[\mathcal{E}].
		\end{align}
		and hence $\Pr[\mathcal{E}]< 1-\epsilon$. Recall that \algone~calls \onerun~for $\lceil\log_{1-\epsilon}\epsilon\rceil$ times in parallel, so the probability of at least one run of \onerun~finds $A_1$ and $A_2$ satisfying $\max\{c(A_1),c(A_2)\}<B-\max\{\epsilon B, c(w)\}$ is no less than $1-(1-\epsilon)^{\log_{1-\epsilon}\epsilon}=1-\epsilon$. Using Lemma~\ref{lm:onerun}, we have
		$\E[f(S)]\geq \frac{(1-\epsilon)(1-2\epsilon)(1-2\alpha-2\epsilon)}{4(1-\epsilon)}f(O)\geq (1/8-\epsilon)f(O)$ when $\alpha=1/4$.
	\end{itemize}
	According to the above discussion, the theorem follows.
\end{proof}

Note that the complexity of \algone~is dominated by Lines~\ref{ln:for_start_skp1}--\ref{ln:for_end_skp1}, where \onerun~is run for multiple times in parallel. Therefore, leveraging Lemma~\ref{lma:complexityofrandbatch}, we can also get:

\begin{theorem} \label{thm:complexityofparskp1}
	The adaptive complexity and query complexity of \algone~are $\mathcal{O}(\log n)$ and $\OO(nr\log^2n)$ respectively, or $\OO(\log n\log r)$ and $\OO(n\log^2n\log r)$ respectively.
\end{theorem}
\begin{proof}
	Note that \onerun~calls \RB~using $M=\lceil \epsilon^{-2}\rceil$ and $I\subseteq \N_1$, and we have $\max_{u,v\in \N_1}\frac{c(u)}{c(v)}\leq \frac{n}{\epsilon}$ due to the definition of $\N_1$. Therefore, according to Lemma~\ref{lma:complexityofrandbatch}, each call of \RB~incurs adaptive complexity of $\OO(\frac{1}{\epsilon}\log \frac{n}{\epsilon}+\frac{1}{\epsilon^2})$ and query complexity of $\OO(\frac{nr}{\epsilon}\log \frac{n}{\epsilon}+\frac{nr}{\epsilon^2})$ if binary search is not used in \RB, or incurs adaptive complexity of $\OO((\frac{1}{\epsilon}\log \frac{n}{\epsilon}+\frac{1}{\epsilon^2})\cdot\log r)$ and query complexity of $\OO((\frac{n}{\epsilon}\log \frac{n}{\epsilon}+\frac{n}{\epsilon^2})\cdot\log r)$ if binary search is used. Besides, note that \algone~calls \onerun~for $\OO(\frac{1}{\epsilon}\log\frac{1}{\epsilon})$ times for every $\rho\in Z$ in parallel, and we have $|Z|= \OO(\frac{1}{\epsilon}\log\frac{n}{\epsilon})$. The $\mathsf{USM}$ algorithm called in Line~\ref{ln:usm_in_skp1} of \algone~or Line~\ref{ln:usm_in_probe} of \onerun~incurs adaptive complexity of $\OO(\frac{1}{\epsilon}\log\frac{1}{\epsilon})$ and query complexity of $\OO(\frac{n}{\epsilon^4}\log^3\frac{1}{\epsilon})$. Combining all the above results, we know that the adaptive complexity and query complexity of \algone~are $\OO(\frac{1}{\epsilon}\log \frac{n}{\epsilon}+\frac{1}{\epsilon^2})=\OO(\log n)$ and $\OO(\frac{1}{\epsilon^2}\log\frac{1}{\epsilon}\log\frac{n}{\epsilon}(\frac{nr}{\epsilon}\log\frac{n}{\epsilon}+\frac{nr}{\epsilon^2}+\frac{n}{\epsilon^4}\log^3\frac{1}{\epsilon}))=\OO(nr\log^2 n)$ respectively, or $\OO((\frac{1}{\epsilon}\log \frac{n}{\epsilon}+\frac{1}{\epsilon^2})\cdot\log r)$  and $\OO(\frac{1}{\epsilon^2}\log\frac{1}{\epsilon}\log\frac{n}{\epsilon}((\frac{n}{\epsilon}\log \frac{n}{\epsilon}+\frac{n}{\epsilon^2})\cdot\log r+\frac{n}{\epsilon^4}\log^3\frac{1}{\epsilon}))=\OO(n\log^2 n\log r)$ respectively.
\end{proof}
\subsection{The Algorithm for the SSP}\label{sec:ssp}
\begin{algorithm}[t]
	\caption{$\mathsf{ParSSP}(p, \epsilon, f(\cdot))$}
	\label{alg:ssp}
 \KwIn{probability $p$, precision $\epsilon\in (0,1)$, and submodular function $f(\cdot)$\;}
	$T\leftarrow\emptyset$; $I\leftarrow\N$; $M\gets  \big\lceil \frac{\log_{1-\epsilon} \frac{\epsilon}{r}+2}{\epsilon^2}\big\rceil$; $\ell\leftarrow \big\lceil\log_{1-\epsilon}\frac{\epsilon}{r}\big\rceil+1$\label{ln:sspfirst}\;
	$u^*\leftarrow\arg\max_{u\in\N} f(u)$; $\rho_{max}\gets f(u^*)$\;
	\For{$i\leftarrow 1$ {\rm\bfseries to} $\ell$\label{ln:for_start_ssp}}{
		$\rho_i\leftarrow\rho_{max}\cdot (1-\epsilon)^{i-1}$\;
		$(A_i,U_i,L_i)\gets \mathsf{RandBatch}(\rho_i,I,M,p,\epsilon,f_T(\cdot))$\label{ln:AUL}\;
		$T\gets T\cup A_i$; $I\gets I\setminus(U_i\cup L_i)$\label{ln:for_end_ssp}\;
	}
	\Return{$S\gets\arg\max_{X\in\{T,\{u^*\}\}}f(X)$\label{ln:return_ssp}}\;
\end{algorithm}
In this section, we introduce \PS~(as shown in Algorithm~\ref{alg:ssp}), which calls \RB~using $\ell$ non-increasing thresholds $\rho_1,\dotsc, \rho_{\ell}$ to find $\ell$ sequences of elements (i.e., $A_1,\dotsc, A_{\ell}$), and then splices them together to get $T$ (Lines~\ref{ln:for_start_ssp}--\ref{ln:for_end_ssp}). Note that the elements in $U_j$ and $L_j$ returned by \RB~(for every $j\in [i]$) are all neglected when seeking for $A_{i+1}$ (Line~\ref{ln:for_end_ssp}), which is useful for the performance analysis presented shortly. The final solution $S$ returned by \PS~is the best one among $T$ and the single element in $\N$ with maximum objective function value (Line~\ref{ln:return_ssp}).

Now we begin to show the performance analysis of \PS, which is more involved than those of \algone, as \PS~calls \RB~with $p<1$ to introduce additional randomness. We first introduce some definitions useful in our analysis. When \PS~finishes, let $\U=\cup_{i=1}^\ell U_i, \LL=\cup_{i=1}^\ell L_i$, $O_{\mathit{small}}=\{u\colon u\in O\setminus (\U\cup \LL)\wedge f(u\mid T)\leq  \rho_{\ell}\}$ and $O_{\mathit{big}}=O\setminus (O_{\mathit{small}}\cup \U\cup \LL)$. So each element $u\in O_{\mathit{big}}$ must satisfy $f(u\mid T)> \rho_{\ell}$ and $T\cup \{u\}\notin\I$. 

Based on the definitions given above, we introduce a mapping $\Upsilon(\cdot)$ for performance analysis in Lemma \ref{lm:mapping}. The intuition of this mapping is to map the elements in $O_{\mathit{big}}$ to those in $T$, so that the utility loss resulting from excluding $O_{\mathit{big}}$ from $T$ can be bounded (as shown by Lemma \ref{lm:boundloss}). 
\begin{lemma} \label{lm:mapping}
	There exists a mapping $\Upsilon: O_{big} \mapsto T$ satisfying: 
	\begin{enumerate}
		\item Each $v\in O_{big}$ can be added into $T$ without violating the $k$-system constraint at the moment that $\Upsilon(v)$ is added into $T$.
		\item Let $\Upsilon^{-1}(u) =\left\{v\in O_{big}: \Upsilon(v)=y \right\}$ for any $u\in T$. Then we have $|\Upsilon^{-1}(u)|\leq k$.
	\end{enumerate}
\end{lemma}
\begin{lemma} \label{lm:boundloss}
	For any $u\in \U$, let $\rho(u)$ denote the threshold used by \PS~when $u$ is considered to be added into $T$ and define $\rho(u)=0$ for other $u\in\N$. For any $u\in \U$, we have 
	\begin{enumerate}
		\item $f(u\mid T)\leq \rho(u)/(1-\epsilon)$;
		\item and if $u\in T\wedge \Upsilon^{-1}(u)\neq \emptyset$, we can get $\forall v\in \Upsilon^{-1}(u)\colon f(v\mid T)\leq \rho(u)/(1-\epsilon)$.
	\end{enumerate}
\end{lemma}
\begin{proof}[Proof of Lemma~\ref{lm:mapping}]
	Let $T=\{u_1,\cdots,u_s \}$ be a set where elements are ordered according to the order that they are added into $T$. Let $Z_s=O_{big}$ and construct the mapping by executing the following iterations from $t=s$ to $t=0$. In iteration $t$, we first compute a set $Q_t=\{x\in Z_t \setminus\{u_1,\cdots,u_{t-1}\}:\{u_1,\cdots,u_{t-1},x\} \in \mathcal I \}$. If $|Q_t| \leq k$, then we set $R_t$ =$Q_t$; otherwise, we select a subset $R_t\subseteq Q_t$ such that $|R_t|=k$. Next, we assign $\Upsilon(u)=u_t$ for each $u\in R_t$ and update $Z_{t-1}=Z_t \setminus R_t$. Then we proceed to iteration $(t-1)$.	
	
	It is clear that the $\Upsilon(\cdot)$ constructed by the above process satisfies Conditions 1-2. Therefore, we only need to show that every $u\in O_{big}$ is mapped to an element in $T$, which is equivalent to showing $Z_0=\emptyset$ since each $u\in Z_s\setminus Z_0$ is mapped to an element in $T$ by the above process. To prove $Z_0=\emptyset$, we use induction and show that $|Z_t| \leq k\cdot t$ for all $0\leq t\leq s$.
	
	\begin{enumerate}
		\item For $t=s$, let $M=T\cup O_{big}$. It is obvious that every element $u\in O_{big}$ satisfies $T \cup \{u\} \notin \I$ by the definition of $O_{big}$. Hence, we infer that $T$ is a base of $M$. Since $O_{big}\in \I$, we obtain $|Z_s|=|O_{big}|\leq k|T|=k \cdot s$ by the definition of $k$-system.
		\item For $0\leq t< s$, assume that $|Z_t|\leq k \cdot t$ for some $t$. If $|Q_t|>k$, then we set $|R_t|=k$ and thus $|Z_{t-1}|=|Z_t|-k\leq k(t-1)$. If $|Q_t|\leq k$, then we observe that no element $u\in Z_{t-1} \setminus \{u_1,\cdots, u_{t-1}\}$ satisfies $\{u_1,\cdots,u_{t-1}\}\cup \{u\}\in \I$ due to the above process for constructing $\Upsilon(\cdotp)$. Now consider the set $M'=\{u_1,\cdots,u_{t-1}\}\cup Z_{t-1}$, we see that $\{u_1,\cdots,u_{t-1}\}$ is a base of $M'$ and $Z_{t-1}\in \I$, which implies $|Z_{t-1}|\leq k (t-1)$ by the definition of $k $-system.
	\end{enumerate}
	By induction, we have shown that $Z_t\leq k\cdot t$ for all $0\leq t \leq s$, which implies $Z_0=\emptyset$. Therefore, the lemma is proved.
\end{proof}
\begin{proof}[Proof of Lemma~\ref{lm:boundloss}]
	The lemma is trivial for $u\in U_{1}$ since $\rho(u)=\rho_{max}$. Next, suppose that $u\in U_i~(i>1)$ and there exists $v\in \Upsilon^{-1}(u)$ such that $f(v\mid T)> \rho(u)/(1-\epsilon)$ (for a contradiction). Let $T'$ be the set of elements in $T$ when $u$ is considered to be added to $T$. Then we get $ T\setminus\cup_{t=i}^{\ell}U_{t}\subseteq T'\subseteq T$, which implies that $v$ can be added to $T\setminus\cup_{t=i}^{\ell}U_{t}$ without violating the $k$-system constraint by Lemma \ref{lm:mapping} and the hereditary property of $k$-system. By submodularity, we obtain $f(v\mid T\setminus\cup_{t=i}^{\ell}U_{t})> \rho(u)/(1-\epsilon)$, and thus $v\in \LL$. But this contradicts the fact that $v\in O_{big}$.
	A similar line of reasoning can demonstrate that $f(u\mid T)\leq \rho(u)/(1-\epsilon)$. Combining all of the above completes the proof.
\end{proof}	
By applying submodularity and Lemma \ref{lm:boundloss}, we obtain the following lemma:
\begin{lemma}\label{lm:sspwithO}
	For any $u\in\N$, let $X_u=1$ if $u\in T$ and $X_u=0$ otherwise; let $Y_u=1$ if $u\in O\cap\U\setminus T$ and $Y_u=0$ otherwise. When \PS~finishes, we have
	\begin{align}
	\nonumber &f(T\cup O)
	\leq f(T)+f(O_{\mathit{small}}\mid T)+ f(O_{\mathit{big}}\mid T)+f(\LL\cap O\mid T)\\
	\nonumber	&+ f(O\cap \U\setminus T\mid T)
	\leq(1+\epsilon)f(S)+\epsilon f(O)+k\sum\limits_{u\in \N}\frac{X_u\cdot \rho(u)}{1-\epsilon}+\sum\limits_{u\in \N}\frac{Y_u\cdot \rho(u)}{1-\epsilon}.
	\end{align}
\end{lemma}
\begin{proof}
	When \PS~finishes, the set $O\setminus T$ can be partitioned into several disjoint subsets: $O_{\mathit{small}}$, $\LL\cap O$, $O_{\mathit{big}}$ and $O\cap \U\setminus T$. By submodularity, we have
	\begin{equation}
	f(O\cup T)-f(T)\leq f(O_{\mathit{small}}\mid T)+ f(O_{\mathit{big}}\mid T)+f(\LL\cap O\mid T)+ f(O\cap \U\setminus T\mid T).\label{eqn:tocombinestart}
	\end{equation}
	Besides, using Lemma~\ref{lm:properties} and submodularity, we have
	\begin{equation}
	f(\LL\cap O\mid T)\leq \sum\nolimits_{i=1}^{\ell} f(L_i\cap O\mid \cup_{j=1}^{i-1} A_i)\leq\epsilon^{-1}\cdot {\ell f(O)}/{ {M}}\leq\epsilon f(O),
	\end{equation}
	where the last inequality is due to $\ell\leq  \log_{1-\epsilon}\frac{\epsilon}{r}+2$ and $M\geq \frac{\log_{1-\epsilon}\frac{\epsilon}{r}+2}{\epsilon^{2}}$ according to Line~\ref{ln:sspfirst} of Algorithm~\ref{alg:ssp}. According to the definition of $O_{\mathit{small}}$, we have
	\begin{equation}
	f(O_{\mathit{small}}\mid T)\leq\sum_{u\in O_{small}}f(u\mid T)\leq r\cdot\rho_{\ell}\leq \epsilon \cdot\rho_{max}\leq \epsilon f(S).
	\end{equation}
	By Lemma \ref{lm:boundloss} and submodularity, we have
	\begin{eqnarray}
	f(O\cap\U\setminus T\mid T)\leq \sum_{u\in O\cap\U\setminus T}f(u\mid T)=\sum_{u\in O\cap\U\setminus T}\rho(u)/(1-\epsilon)=\sum_{u\in \N}Y_u\cdot \rho(u)/(1-\epsilon).\label{eqn:OcapU}
	\end{eqnarray}
	and
	\begin{align}
	\nonumber &f(O_{\mathit{big}}\mid T)\leq \sum_{u\in O_{big}}f(u\mid T)=\sum_{u\in T}\sum_{v\in\Upsilon^{-1}(u)}f(v\mid T)\\
	\nonumber &\leq \sum_{u\in T}\sum_{v\in\Upsilon^{-1}(u)}\rho(u)/(1-\epsilon)\leq \sum_{u\in T}k\cdot\rho(u)/(1-\epsilon)= k\sum_{u\in \N}X_u\cdot \rho(u)/(1-\epsilon).	
	\end{align}	
	Combining all of the above completes the proof.	
\end{proof}
Note that both $\rho(u)$ and $\Upsilon^{-1}(u)$ are random for any $u\in\N$, and the randomness is caused by both
Line \ref{ln:random2-1} of Algorithm \ref{alg:rb} and the random selection in the \RQ~function. So we study their expectation and get the following lemma:
\begin{lemma} \label{lm:boundXuYu}
	We have
	\begin{eqnarray}
	\nonumber&&\E[k\sum_{u\in \N}X_u\cdot\rho(u)+\sum\limits_{u\in \N}{Y_u\cdot \rho(u)}]\leq (k+\frac{1-p}{p})(1-\epsilon)^{-2}\cdot{\E[f(T)]}
	\end{eqnarray}
\end{lemma}
\begin{proof}
	Recall that $\mathcal{U}=\cup_{i=1}^\ell U_i$, where $U_i$ is generated in Line~\ref{ln:AUL} of Algorithm~\ref{alg:ssp}. Suppose that $|\mathcal{U}|=h$, so we can create a random sequence $\{u_1,\dotsc, u_h, u_{h+1},\dotsc, u_{n}\}$, where $\{u_1,u_2,\dotsc, u_h\}$ are the elements in $\mathcal{U}$ listed according to the order that they are selected, and $\{u_{h+1},\dotsc, u_{n}\}$ are the elements in $\N\setminus \U$ listed in an arbitrary order. For any $i\in [n]$, let $\delta(u_i)=f(u_i\mid \{u_1,\dotsc,u_{i-1}\}\cap T\cup \lambda(u_i))$ if $u_i\in T$ and otherwise $\delta(u_i)=0$, where the function $\lambda(\cdot)$ has been defined in Lemma~\ref{lm:density}. So $f(T)\geq\sum_{i=1}^{n}\delta(u_i)$, and hence we only need to prove
	\begin{equation}
	\nonumber	\forall i\in [n]\colon  \E[k\cdot X_{u_i}\cdot\rho(u_i)+Y_{u_i}\cdot\rho(u_i)]\leq (k+\frac{1-p}{p})(1-\epsilon)^{-2}\cdot \E[\delta(u_i)],
	\end{equation}
	due to the linearity of expectation. Let $\mathcal{F}_{i-1}$ be the filtration capturing all the random choices made by \PS~until the moment right before selecting $u_i$. According to the law of total expectation, it is sufficient to prove
	\begin{eqnarray}
	\nonumber\forall i\in [n], \forall \mathcal{F}_{i-1}\colon  &&\E[k\cdot X_{u_i}\cdot\rho(u_i)+Y_{u_i}\cdot\rho(u_i)\mid \mathcal{F}_{i-1}]\\
	&&\leq (k+\frac{1-p}{p})(1-\epsilon)^{-2}\cdot \E[\delta(u_i)\mid \mathcal{F}_{i-1}].\label{eqn:expfiltration}
	\end{eqnarray}
	Note that $\mathcal{F}_{i-1}$ determines whether $u_i\in \U$. Therefore, according to the definitions of $\rho(\cdot)$ and $\delta(\cdot)$, Eqn.~\eqref{eqn:expfiltration} trivially holds under the case of $u_i\notin \U$ given $\mathcal{F}_{i-1}$. So in the sequel, we only consider the case of $u_i\in \U$ given $\mathcal{F}_{i-1}$.
	By similar reasoning with that in Lemma~\ref{lm:properties}, we can get:
	\begin{equation}
	\nonumber	\E[\delta(u_i)\mid \mathcal{F}_{i-1}]= p\cdot \E [f(u_i\mid \{u_1,\dotsc,u_{i-1}\}\cap T\cup \lambda(u_i))\mid \mathcal{F}_{i-1}].
	\end{equation}
	Note that $\rho(u_i)$ is deterministic given $\mathcal{F}_{i-1}$, then we have
	\begin{eqnarray}
	\nonumber &&\E[k\cdot X_{u_i}\cdot\rho(u_i)+Y_{u_i}\cdot\rho(u_i)\mid \mathcal{F}_{i-1}]=\rho(u_i)\cdot\E[k\cdot X_{u_i}+Y_{u_i}\mid \mathcal{F}_{i-1}]\\
	\nonumber &&\leq \rho(u_i)(k\cdot p+1-p)
	\end{eqnarray}		
	where the inequality is due to the reason that $u$ is accepted with probability of $p$ and discarded with probability of $1-p$. Besides, by similar reasoning with that in Lemma~\ref{lm:density}, we can get
	\begin{equation}
	\nonumber	\E [f(u_i\mid \{u_1,\dotsc, u_{i-1}\}\cap T\cup \lambda(u_i))\mid \mathcal{F}_{i-1}] \geq (1-\epsilon)^2\rho(u_i).
	\end{equation}
	The proof is now complete by combining the above.
\end{proof}
Using Lemmas \ref{lm:sspwithO}–\ref{lm:boundXuYu}, we get the performance bounds of \PS~as follows:
\begin{theorem}\label{thm:k-system-ratio}
	For the non-monotone SSP, \PS~algorithm can return a solution $T$ satisfying $\mathbb{E}[f(T)]\geq (1-\epsilon)^5{(\sqrt{k+1}+1)^{-2}}f(O)$ by setting $p=(1+\sqrt{k+1})^{-1}$. The adaptive complexity and query complexity of \PS~are $\mathcal{O}(\sqrt{k}\log^2 n)$ and $\OO(\sqrt{k}nr\log^2n)$ respectively, or $\OO(\sqrt{k}\log^2 n\log r)$ and $\OO($
	$\sqrt{k}n\log^2n\log r)$ respectively.
\end{theorem}		
\begin{proof}
	We first quote the following lemma presented in \shortcite{buchbinder2014submodular}:
	\begin{lemma}[\shortcite{buchbinder2014submodular}]\label{lm:buchbinder}
		Given a ground set $\N$ and any non-negative submodular function $g(\cdot)$ defined on $2^{\N}$, we have $\E[g(Y)]\geq (1-p)g(\emptyset)$ if $Y$ is a random subset of $\N$ such that each element in $\N$ appears in $Y$ with probability of
		at most $p$ (not necessarily independently).
	\end{lemma}
	By combining Lemmas \ref{lm:sspwithO}–\ref{lm:boundXuYu} and using the fact that $f(T)\leq f(S)$, we can get
	\begin{equation}
	\nonumber \E[f(T\cup O)]\leq (1+\epsilon)\E[f(S)]+\epsilon f(O)+(k+\frac{1-p}{p})(1-\epsilon)^{-3}\cdot{\E[f(S)]}.
	\end{equation} 
	Let $g(\cdot)=f(\cdot\cup O)$, then we can use Lemma~\ref{lm:buchbinder} to get $\E[f(T\cup O)]\geq (1-p)f(O)$. Combining this with the above equation, and setting $p=(1+\sqrt{k+1})^{-1}$, we get
	\begin{equation}
	\nonumber\frac{\E[f(S)]}{f(O)}\geq \frac{1-p-\epsilon}{1+\epsilon+(k+\frac{1-p}{p})(1-\epsilon)^{-3}}\geq (1-\epsilon)^{5}{(\sqrt{k+1}+1)^{-2}}
	\end{equation}
	when $\epsilon\in[0,0.4)$, which completes the proof on the approximation ratio.
	
	In a manner similar to the proof of Theorem \ref{thm:complexityofparskp1}, we can employ Lemma \ref{lma:complexityofrandbatch} to analyze and determine the complexity of \PS. Then combining all of the above completes the proof.
\end{proof}
\section{Extensions for Cardinality Constraint}
\label{sec:smc}
As cardinality constraint is a special case of knapsack constraint and $k$-system constraint (where $k=1$), our \algone~and \PS~algorithms can be directly applied to the non-monotone SMC, for which the performance bounds shown in Theorem~\ref{thm:boundofalgone} and Theorem~\ref{thm:k-system-ratio} still hold. Interestingly, by a more careful analysis, we find that \PS~actually achieves a better approximation ratio for the SMC, while its complexities remain the same. This result is shown in Theorem~\ref{thm:boundofsspforsmc}. We roughly explain the reason as follows. Since the cardinality constraint is more restrictive than the $k$-system constraint, we can use a more effective mapping to bound the utility loss caused by the elements in $O_{big}$ by using only the elements in $T\setminus O$ instead of all the elements in $T$. This leads to stronger versions of Lemma \ref{lm:sspwithO} and Lemma \ref{lm:boundXuYu}, as shown by Lemma \ref{lm:smc1} and Lemma \ref{lm:smc2}.
\begin{lemma}\label{lm:smc1}
	For any $u\in\N$, let $X_u^\prime=1$ if $u\in T\setminus O$ and $X_u^\prime=0$ otherwise; let $Y_u^\prime=1$ if $u\in O\cap\U\setminus T$ and $Y_u^\prime=0$ otherwise. When \PS~finishes, we have
	\begin{eqnarray}
	\nonumber f(T\cup O)\leq(1+\epsilon)f(S)+\epsilon f(O)+\sum\limits_{u\in \N}\frac{X_u^\prime\cdot \rho(u)}{1-\epsilon}+\sum\limits_{u\in \N}\frac{Y_u^\prime\cdot \rho(u)}{1-\epsilon}.
	\end{eqnarray}
\end{lemma}
\begin{proof}
	Using similar reasoning as Eqn. \eqref{eqn:tocombinestart}-\eqref{eqn:OcapU}, we can get
	\begin{eqnarray}
	\nonumber f(T\cup O)\leq(1+\epsilon)f(S)+\epsilon f(O)+f(O_{\mathit{big}}\mid T)+\sum\limits_{u\in \N}\frac{Y_u^\prime\cdot \rho(u)}{1-\epsilon}.
	\end{eqnarray}
	Based on the cardinality constraint property, any element in $O_{big}$ can be added to the candidate solution $T$ without violating the constraint, provided that it is added prior to $u_{last}$, where $u_{last}$ denotes the last added element in $T$. So according to the definition of $O_{big}$ and submodularity, we must have 
	\begin{enumerate}
		\item if $|T|<r$, then $|O_{big}|=0$;
		\item if $|T|=r$, then $\forall u\in O_{big}: f(u\mid T)\leq \rho(u_{last})/(1-\epsilon)$ and $|O_{big}|\leq |O\setminus \U|\leq |O|-|O\cap T|\leq r-|O\cap T|=|T|-|T\cap O|=|T\setminus O|$.
	\end{enumerate}
	Thus, we can get
	\begin{eqnarray}
	\nonumber &&f(O_{big}\mid T)\leq \sum_{u\in O_{big}}f(u\mid T)\leq \sum_{u\in O_{big}}\rho(u_{last})/(1-\epsilon)\\
	\nonumber&\leq&\sum_{u\in T\setminus O}\rho(u_{last})/(1-\epsilon)\leq \sum_{u\in T\setminus O}\rho(u)/(1-\epsilon)=\sum_{u\in \N}X_u^\prime\cdot \rho(u)/(1-\epsilon)
	\end{eqnarray}
	The proof now completes by combining the above.
\end{proof}
\begin{lemma}\label{lm:smc2}
	We have
	\begin{eqnarray}
	\nonumber&&\E[\sum_{u\in \N}X_u^\prime\cdot\rho(u)+\sum\limits_{u\in \N}{Y_u^\prime\cdot \rho(u)}]\leq (1-\epsilon)^{-2}\cdot\frac{\max\{p,1-p\}}{p}\cdot{\E[f(T)]}
	\end{eqnarray}
\end{lemma}
\begin{proof}
	By following a similar argument as in the proof of Lemma \ref{lm:boundXuYu}, we only need to prove
	\begin{eqnarray}
	\nonumber\forall i\in [n], \forall \mathcal{F}_{i-1}\colon  &&\E[ X_{u_i}^\prime\cdot\rho(u_i)+Y_{u_i}^\prime\cdot\rho(u_i)\mid \mathcal{F}_{i-1}]\\
	\nonumber	&&\leq (1-\epsilon)^{-2}\cdot\frac{\max\{p,1-p\}}{p}\cdot \E[\delta(u_i)\mid \mathcal{F}_{i-1}].
	\end{eqnarray}
	under the case of $u_i\in\U$ given $\mathcal{F}_{i-1}$ to prove this lemma. Consider the following two scenarios.
	\begin{enumerate}
		\item $u_i\in O$. In this case, we must have $X_{u_i}^\prime=0$, and $Y_{u_i}^\prime=1$ if ${u_i}$ is discarded by the algorithm. Therefore, $\E[ X_{u_i}+Y_{u_i}\mid \mathcal{F}_{i-1}]=1-p$.
		\item $u\notin O$: In this case, we must have $Y_{u_i}^\prime=0$, and $X_{u_i}^\prime=1$ if ${u_i}$ is not discarded by the algorithm. Therefore, $\E[ X_{u_i}+Y_{u_i}\mid \mathcal{F}_{i-1}]=p$.
	\end{enumerate}
	Since $\rho(u_i)$ is deterministic given $\mathcal{F}_{i-1}$, we have
	\begin{eqnarray}
	\nonumber \E[X_{u_i}\cdot\rho(u_i)+Y_{u_i}\cdot\rho(u_i)\mid \mathcal{F}_{i-1}]=\rho(u_i)\cdot\E[ X_{u_i}+Y_{u_i}\mid \mathcal{F}_{i-1}]\leq \rho(u_i)\cdot {\max\{p,1-p\}}.
	\end{eqnarray}	
	Recall that		
	\begin{eqnarray}
	\nonumber \E[\delta(u_i)\mid \mathcal{F}_{i-1}]=p\cdot \E [f(u_i\mid \{u_1,\dotsc,u_{i-1}\}\cap T\cup \lambda(u_i))\mid \mathcal{F}_{i-1}]\geq (1-\epsilon)^2\cdot p\cdot \rho(u_i).
	\end{eqnarray}
	The proof now completes by combining the above.
\end{proof}

Using Lemma \ref{lm:smc1} and Lemma \ref{lm:smc2}, we can improve the approximation ratio of \PS~for the non-monotone SMC, as stated by Theorem \ref{thm:boundofsspforsmc}.
\begin{theorem} \label{thm:boundofsspforsmc}
	For the non-monotone SMC, \PS~can return a solution $S$ satisfying $\E[f(S)]\geq (1/4-\epsilon)\OPT$ by setting $p=1/2$, under the same adaptivity and query complexity as those shown in Theorem~\ref{thm:k-system-ratio} where $k=1$.
\end{theorem}
\begin{proof}
	By combining Lemmas \ref{lm:smc1}--\ref{lm:smc2} and Lemma \ref{lm:buchbinder} and using the fact that $f(T)\leq f(S)$, we get
	\begin{equation}
	\nonumber (1-p)f(O)\leq (1+\epsilon)\E[f(S)]+\epsilon f(O)+(1-\epsilon)^{-3}\cdot\frac{\max\{p,1-p\}}{p}\cdot{\E[f(S)]},
	\end{equation} 
	By setting $p=1/2$ and rearranging the above inequality, we derive the approximation ratio of \PS~as
	\begin{eqnarray}
	\nonumber \frac{\E[f(S)]}{f(O)}\geq \frac{1/2-\epsilon}{1+\epsilon+(1-\epsilon)^{-3}}\geq \frac{1}{4}-\epsilon.
	\end{eqnarray}
	Finally, the complexity analysis is the same with that in Theorem~\ref{thm:k-system-ratio} where $k=1$.
\end{proof}
\section{Performance Evaluation}\label{sec:experiment}
In this section, we evaluate the performance of our algorithm \algone~(resp. \PS) by comparing it with state-of-the-art algorithms for the non-monotne SKP (resp. SSP). The evaluation metrics include both the objective function value (i.e., utility) and the number of oracle queries to the objective function. We conduct experiments on three real-world applications, which are described in detail below.
\subsection{Applications}\label{sec:app}
\noindent\textbf{Revenue Maximization.} This application is also considered in~\shortcite{mirzasoleiman2016fast,balkanski2018non,fahrbach2019non,han2021approximation,cui2021randomized,amanatidis2021submodular,amanatidis2022fast}. In this problem, we are given a social network $G=(\N,E)$ where each node $u\in\N$ represents a user with a cost $c(u)$, and each edge $(u,v)\in E$ has a weight $w_{u,v}$ denoting the influence of $u$ on $v$. There are also $t$ advertisers for different products; each advertiser $i\in [t]$ needs to select a subset $S_i\subseteq V$ of seed nodes and provide a sample of product $i$ to each user $u\in S_i$ (and also pay $c(u)$ to $u$) for advertising product $i$. Following \shortcite{mirzasoleiman2016fast,balkanski2018non,fahrbach2019non,amanatidis2021submodular,amanatidis2020fast}, we define the revenue of advertiser $i$ as $f_i(S_i)=\sum_{u\in \N\setminus S_i}\sqrt{\sum_{v\in S_i}w_{v,u}}$ (which measures the total influence of the seed nodes on the non-seed nodes) and define the objective function of the application as $\sum_{i\in [t]}f_i(S_i)$ (i.e., the total revenue of all advertisers). \bl{This function has already been shown by \shortcite{mirzasoleiman2016fast,amanatidis2020fast,balkanski2018non,fahrbach2019non,amanatidis2021submodular} to be a non-monotone submodular function.} We also follow the existing work to consider the following constraints to better model the demands of real-world scenarios: 1) each node can serve as a seed node for at most $q$ products; 2) the total number of product samples available for each product is at most $m$; and 3) the total seed cost is bounded by $B$ (i.e., $\sum_{i\in [t]}\sum_{u\in S_i}c(u)\leq B$). The first two constraints constitute a $k$-system constraint (where $k=2$), which has been demonstrated by \shortcite{mirzasoleiman2016fast,cui2021randomized}. The third constraint is a conventional knapsack constraint. The edge weights are randomly sampled from the continuous uniform distribution $\mathcal{U}(0,1)$, and the cost of any node $u\in\N$ is defined as $c(u)=g(\sqrt{\sum_{(u,v)\in E}w_{u,v}})$, where $g(x)=1-e^{-\mu}$ is the exponential cumulative distribution function and $\mu$ is set to $0.2$.\\
\textbf{Image Summarization.} This application is also considered in~\shortcite{mirzasoleiman2016fast,balkanski2018non,fahrbach2019non,han2021approximation,cui2024fairness}, the goal is to select a representative set $S$ of images from $\N$. Following \shortcite{mirzasoleiman2016fast,balkanski2018non,fahrbach2019non}, we use a non-monotone submodular function $f(\cdot)$ that can capture both coverage and diversity of $S$ to measure its quality:
\begin{equation*}
f(S)=\sum\nolimits_{u\in \N}\max_{v\in S}s_{u,v}-\frac{1}{|\N|}\sum\nolimits_{u\in S}\sum\nolimits_{v\in S}s_{u,v},
\end{equation*}
where $s_{u,v}$ denotes the cosine similarity between image $u$ and image $v$. Thus, this application is in fact a non-monotone submodular maximization problem. Following~\shortcite{mirzasoleiman2016fast,han2021approximation}, the cost $c(u)$ of any image $u$ is chosen in proportional to the standard deviation of its pixel intensities, such that we assign higher costs to images with higher contrast and lower costs to blurry images. The costs of all images are normalized such that the average cost is $1$. In addition, we restrict the number of images in $S$ that belong to each category to no more than $q$, and we limit the total number of images in $S$ to no more than $m$. It has been indicated in~\shortcite{mirzasoleiman2016fast} that such a constraint is a matroid (i.e., $1$-system) constraint.\\
\textbf{Movie Recommendation.}
This application is also considered in~\shortcite{mirzasoleiman2016fast,feldman2017greed,feldman2020simultaneous,haba2020streaming,badanidiyuru2020submodular,amanatidis2022fast,cui2023constrained}. In this problem, we consider a set $\N$ of movies, each labeled by several genres chosen from a predefined set $G$, and aim to recommend a list of high-quality and diverse movies to a user based on the ratings from similar users. Each movie $u\in\N$ is associated with a 25-dimensional feature vector $q_u$ calculated from user ratings. Following \shortcite{mirzasoleiman2016fast,feldman2017greed,haba2020streaming,amanatidis2020fast,balkanski2018non,fahrbach2019non}, we define the utility of any $S\subseteq \N$ as 
\begin{equation*} 
f(S)=\sum\nolimits_{u\in S}\sum\nolimits_{v\in \N}s_{u,v}-\sum\nolimits_{u\in S}\sum\nolimits_{v\in S}s_{u,v}, 
\end{equation*} 
where we use $s_{u,v}=e^{-\lambda \mathrm{dist}(q_u,q_v)}$ to measure the similarity between movies $u$ and $v$. Thus, this application is in fact a non-monotone submodular maximization problem. Here, $\mathrm{dist}(q_u, q_v)$ is the Euclidean distance between $q_u$ and $q_v$, and $\lambda$ is set to 2. Following~\shortcite{haba2020streaming}, we define the cost $c(u)$ of any movie $u$ to be proportional to $10-r_u$, where $r_u$ denotes the rating of movie $u$ (ranging from 0 to 10), and the costs of all movies are normalized such that the average movie cost is 1. Thus, movies with higher ratings have smaller costs, and we require $\sum_{u\in S}c(u)\leq B$ to ensure that the movies in $S$ have high ratings. Moreover, we also consider the constraint that the number of movies in $S$ labeled by genre $g$ is no more than $m_g$ for all $g\in G$, and that $|S|\leq m$, where $m_g:g\in G$ and $m$ are all predefined integers. It has been indicated in~\shortcite{mirzasoleiman2016fast,feldman2017greed,haba2020streaming} that such a constraint is essentially a $k$-system constraint with $k=|G|$.
\subsection{Experiments for the SKP}\label{sec:exp-skp}
\begin{figure}[t]
	\begin{center}
		\centerline{\includegraphics[width=1\linewidth]{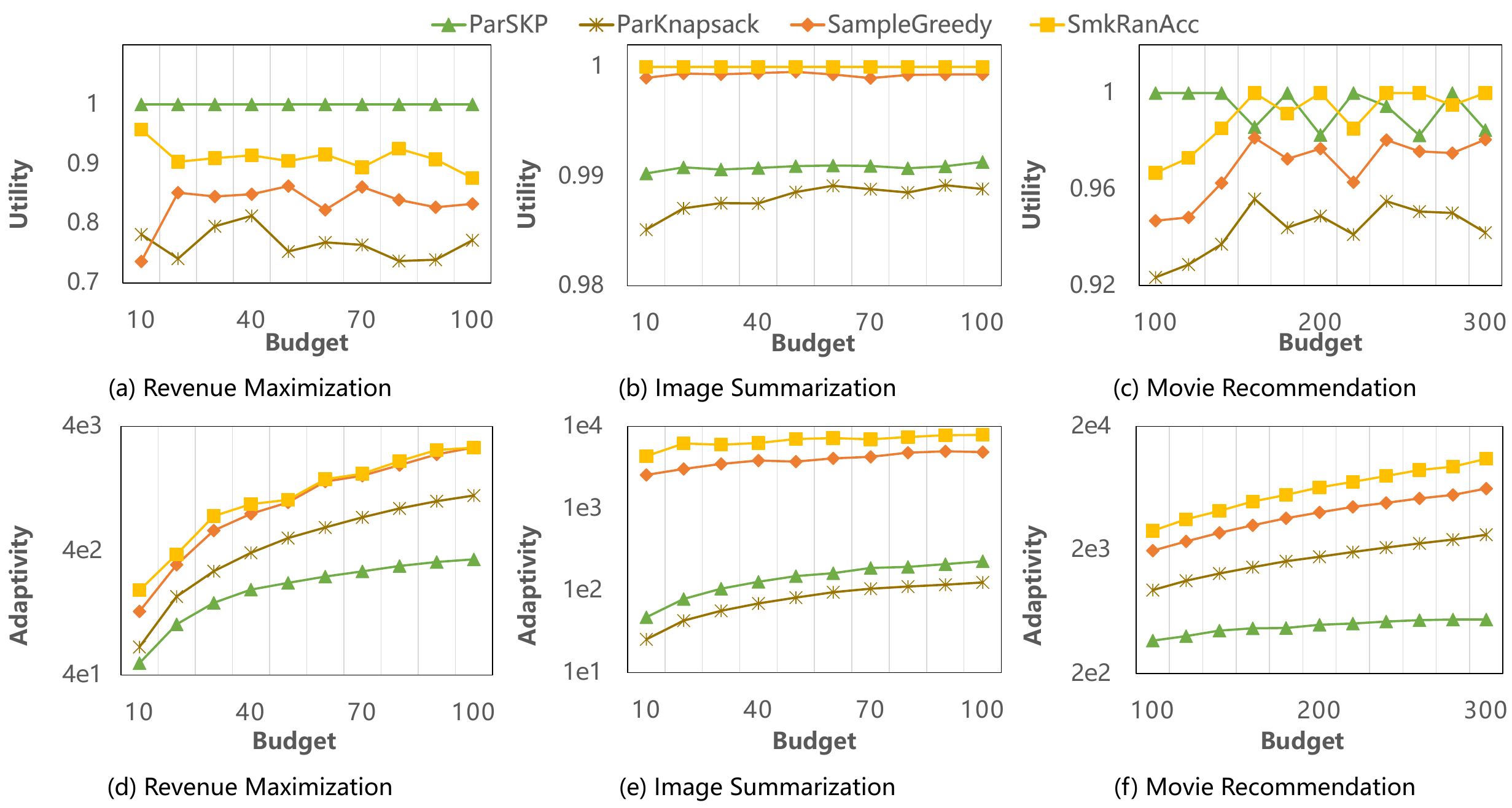}}
		\caption{The figure compares the implemented algorithms on utility and adaptivity, where the plotted utilities are normalized by the largest utility achieved by all algorithms.}
		%
		\label{fig:exp-knapsack}
	\end{center}
\end{figure}
In this section, we compare our \algone~algorithm with several state-of-the-art \textit{practical} algorithms for the non-monotone SKP. 
Specifically, we implement four algorithms: (1) \algone~(Algorithm~\ref{alg:skp1}) using binary search; (2) \textsf{ParKnapsack} \shortcite{amanatidis2023submodular} using binary search;
(3) \textsf{SampleGreedy} \shortcite{amanatidis2022fast}, which gives a $(3-2\sqrt{2}-\epsilon)$-approximation using $\OO(\frac{n}{\epsilon}\log\frac{n}{\epsilon})$ adaptive rounds and queries, implemented using \textit{lazy evaluation}~\shortcite{minoux1978accelerated}; (4) \textsf{SmkRanAcc} \shortcite{han2021approximation}, which gives a $(0.25-\epsilon)$-approximation using $\OO(\frac{n}{\epsilon}\log\frac{r}{\epsilon})$ adaptive rounds and queries. Note that both \textsf{SampleGreedy} and \textsf{SmkRanAcc} are non-parallel algorithms with super-linear adaptivity, so we use these two baselines only to see how other algorithms can approach them. For all the algorithms tested, the accuracy parameter $\epsilon$ is set to $0.1$. Each randomized algorithm is executed independently for $10$ times, and the average result is reported. For the fairness of comparison, we follow~\shortcite{amanatidis2023submodular,amanatidis2021submodular} to use the algorithm in \shortcite{feige2011maximizing} achieving 1/4-approximation and $\OO(1)$ adaptivity for the $\mathsf{USM}$ algorithm. All experiments are run on a Linux server with Intel Xeon Gold 6126 @ 2.60GHz CPU and 256GB memory. We transform the three real-world applications introduced in Section \ref{sec:app} into SKPs and evaluate all implemented algorithms on these problems. The modifications we made and additional experimental settings are as follows:
\begin{itemize}
	\item Revenue Maximization. By setting $t=1$ and $q=m=\infty$, we remove the $k$-system constraint from this application and transform it into the SKP considered by \shortcite{amanatidis2021submodular,amanatidis2022fast,canhv2023linear,han2021approximation}. Following~\shortcite{mirzasoleiman2016fast}, we use the top 5,000 communities of the YouTube network \shortcite{leskovec2014snap} to construct the network $G$, which contains 39,841 nodes and 224,235 edges.
	\item Image Summarization. By setting $q=m=\infty$, we remove the $k$-system constraint from this application and transform it into the SKP considered by \shortcite{canhv2023linear,han2021approximation,cui2022streaming}. Following~\shortcite{balkanski2018non,fahrbach2019non,han2021approximation}, we randomly select $1{,}000$ images from the CIFAR-10 dataset~\shortcite{krizhevsky2009learning} to construct $\N$.
	\item Movie Recommendation. By setting $m=\infty$ and $\forall g\in G:m_g=\infty$, we remove the $k$-system constraint from this application and transform it into the SKP considered by \shortcite{amanatidis2022fast}. In our experiments, We use the MovieLens dataset~\shortcite{badanidiyuru2020submodular,haba2020streaming} which contains 1,793 movies from three genres “Adventure”, “Animation” and “Fantasy”.
\end{itemize}
\textbf{Experimental Results.} In Fig.\ \ref{fig:exp-knapsack}(a)--(c), we compare the implemented algorithms on utility, and the results show \algone~can even achieve better utility compared to non-parallel algorithms \textsf{SampleGreedy} and \textsf{SmkRanAcc}, with the average performance gains of $5\%$ and $3\%$, respectively. Besides, Fig.\ \ref{fig:exp-knapsack}(a)--(c) also show that \algone~achieves significantly better utility than \textsf{ParKnapsack} (with the performance gains of up to $35.74\%$), which is the only existing low-adaptivity algorithm for non-monotone SKP with sub-linear adaptivity and practical query complexity.
In Fig.\ \ref{fig:exp-knapsack}(d)--(f), we compare the implemented algorithms on adaptivity, and the results show that \algone~generally outperforms all the baselines. Specifically, \algone~incurs $4$--$92$ times fewer adaptive rounds than \textsf{SmkRanAcc}, and $3$--$54$ times fewer adaptive rounds than \textsf{SampleGreedy}, which demonstrates the effectiveness of our approach.
\subsection{Additional Experiments for the SKP}\label{sec:exp-skp2}
\begin{figure}[t]
	\begin{center}
		\centerline{\includegraphics[width=\linewidth]{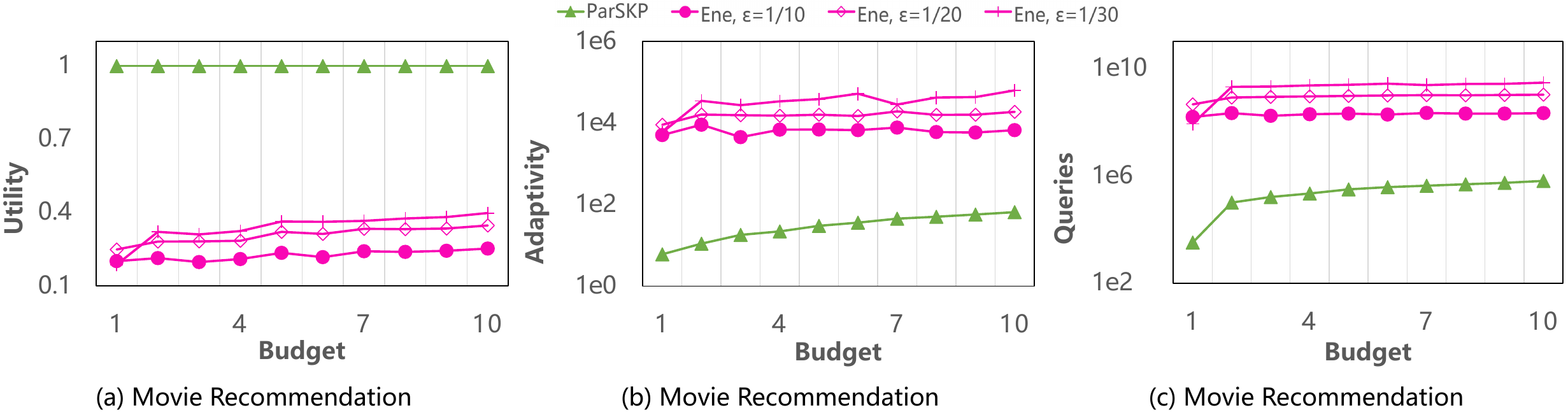}}
		\caption{Comparison of our algorithms with the \textsf{ENE} algorithm~\shortcite{ene2019submodular} on a small instance of movie recommendation. Similar to Figure~\ref{fig:exp-knapsack}, the plotted utilities are normalized by the best utility achieved by the implemented algorithms.}
		\label{fig:exp-ene}
	\end{center}
\end{figure}
Besides the algorithms mentioned in Section~\ref{sec:exp-skp}, we note that \shortcite{ene2019submodular} also propose an algorithm (denoted by \textsf{ENE} for convenience) with $\OO(\log^2 n)$ adaptivity. However, as indicated by \shortcite{amanatidis2021submodular,fahrbach2019non}, the \textsf{ENE} algorithm (based on multi-linear extension) has large query complexity impractical for large datasets, so we compare our \algone~algorithm with \textsf{ENE} using a small instance of the movie recommendation application, under the same settings as those described in Section~\ref{sec:exp-skp} except that the ground set size is smaller ($n=80$). For the \textsf{ENE} algorithm, we use 5,000 samples to simulate an oracle for $F(\cdot)$ or $\nabla F(\cdot)$ (i.e., the multi-linear extension of $f(\cdot)$ and its gradient).


As shown by the experimental results in Fig.~\ref{fig:exp-ene}, our \algone~algorithm outperforms \textsf{ENE} significantly in terms of utility, adaptivity and the number of oracle queries to the objective function, due to the reason that: \textsf{ENE} has a larger theoretical adaptivity of $\OO(\log^2 n)$ than \algone, and its practical performance is not much better than its worst-case theoretical bound due to its design, while evaluating the multilinear extensions in \textsf{ENE} incurs significantly larger number of oracle queries than our algorithms. Moreover, as the approximation ratio of \textsf{ENE} (i.e., $(1+\epsilon)(e^{1+10\epsilon})$) deteriorates quickly with the increasing of $\epsilon$, we have used smaller values of $\epsilon$ (i.e., $\epsilon=1/20; \epsilon=1/30$) to further test its utility. The experimental results in Fig.~\ref{fig:exp-ene} show that our \algone~algorithm consistently outperforms \textsf{ENE} on all these settings, which demonstrate the effectiveness of our approach again.
\subsection{Experiments for the SSP}\label{sec:exp-ssp}
\begin{figure}[!t]
	\begin{center}
\centerline{\includegraphics[width=1\linewidth]{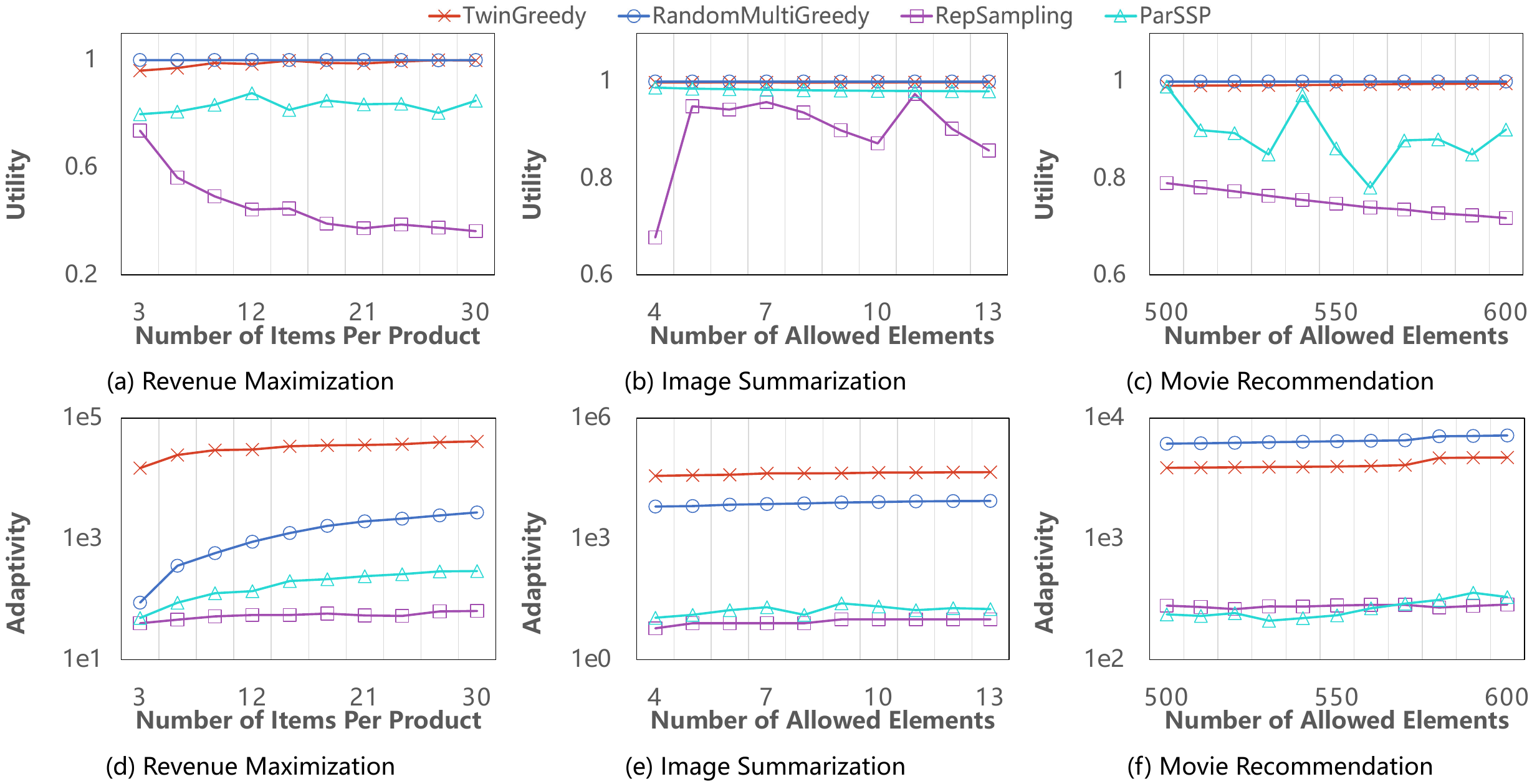}}
		\caption{The figure compares the implemented algorithms on utility and adaptivity, where the plotted utilities are normalized by the largest utility achieved by all algorithms}
		\label{fig:exp-ksystem}
	\end{center}
\end{figure}
In this section, we compare our \PS~algorithm with several state-of-the-art algorithms for the non-monotone SSP. Specifically, we implement four algorithms: (1) \PS~(Algorithm~\ref{alg:ssp}) using binary search; (2) \textsf{RepSampling} \shortcite{quinzan2021adaptive}, which may not provide a constant approximation as we explained in~Appendix \ref{app:kuhnle_error};
(3) \textsf{TwinGreedy} \shortcite{han2020deterministic}, which gives a $((2k+2)^{-1}-\epsilon)$-approximation using $\OO(\frac{n}{\epsilon}\log\frac{r}{\epsilon})$ adaptive rounds and queries; (4) \textsf{RandomMultiGreedy} \shortcite{cui2021randomized}, which gives a $(1+\epsilon)^{-1}(\sqrt{k}+1)^{-2}$-approximation using $\OO(\frac{n}{\epsilon}\log\frac{r}{\epsilon})$ adaptive rounds and queries. Note that both \textsf{TwinGreedy} and \textsf{RandomMultiGreedy} are non-parallel algorithms with super-linear adaptivity, so we use these two baselines only to see how other algorithms can approach them. For all the algorithms tested, the accuracy parameter $\epsilon$ is set to $0.4$. Each randomized algorithm is executed independently for $10$ times, and the average result is reported. All experiments are run on a Linux server with Intel Xeon Gold 6126 @ 2.60GHz CPU and 256GB memory. We transform the three real-world applications introduced in Section \ref{sec:app} into SSPs and evaluated all implemented algorithms on these problems. The modifications we made and additional experimental settings are as follows:
\begin{itemize}
	\item Revenue Maximization. By setting $\forall u\in\N:c(u)=1$ and $B=\infty$, we remove the knapsack constraint from this application and transform it into the SSP considered by \shortcite{cui2021randomized}. Following~\shortcite{balkanski2018non,fahrbach2019non,han2021approximation}, we randomly select $25$ communities from the top $5,000$ communities in the YouTube social network \shortcite{leskovec2014snap} to construct the network $G$, which contains 1,179 nodes and 3,495 edges. We set $t=5$ and $q=2$ while scaling the number of items available for seeding (i.e., $m$) to compare the performance of all algorithms.
	\item Image Summarization. By setting $\forall u\in\N:c(u)=1$ and $B=\infty$, we remove the knapsack constraint from this application and transform it into the SSP considered by \shortcite{cui2021randomized}. Following~\shortcite{mirzasoleiman2016fast,cui2021randomized}, we use the CIFAR-10 dataset~\shortcite{krizhevsky2009learning} to construct $\N$, and restrict the selection of images from three categories: Airplane, Automobile and Bird. We set $q=5$ while scaling $m$ to compare the performance of all algorithms.
	\item Movie Recommendation. By setting $\forall u\in\N:c(u)=1$ and $B=\infty$, we remove the knapsack constraint from this application and transform it into the SSP considered by \shortcite{feldman2017greed,feldman2020simultaneous,haba2020streaming,cui2021randomized}. We use the MovieLens dataset~\shortcite{badanidiyuru2020submodular,haba2020streaming} which contains 1,793 movies from three genres “Adventure”, “Animation” and “Fantasy”, and thus we have $k=3$ in our experiments. We scale $m$ to compare the performance of all algorithms.
\end{itemize}
\textbf{Experimental Results.} In Fig.\ \ref{fig:exp-ksystem}(a)--(c), the utility performance of \textsf{ParSSP} is slightly weaker than the two non-parallel algorithms \textsf{TwinGreedy} and \textsf{RandomMultiGreedy} (with the average performance loss of $10\%$). Besides, Fig.\ \ref{fig:exp-ksystem}(a)--(c) also show that \PS~achieves significantly better utility than \textsf{RepSampling} (with the performance gains of up to $134\%$), which is the only existing low-adaptivity algorithm for non-monotone SKP.
In Fig.\ \ref{fig:exp-ksystem}(d)--(f), we compare the implemented algorithms on adaptivity, and the results indicate that our algorithm \PS~performs comparably to \textsf{RepSampling} and significantly better than non-parallel algorithms \textsf{TwinGreedy} and \textsf{RandomMultiGreedy}. Specifically, \PS~incurs $13$--$3327$ times fewer adaptive rounds than \textsf{TwinGreedy}, and $2$--$582$ times fewer adaptive rounds than \textsf{RandomMultiGreedy}, which demonstrates the effectiveness of our approach.
\subsection{Experiments for the SMC}\label{sec:exp-smc}
\begin{figure}[!ht]
	\begin{center}
\centerline{\includegraphics[width=1\linewidth]{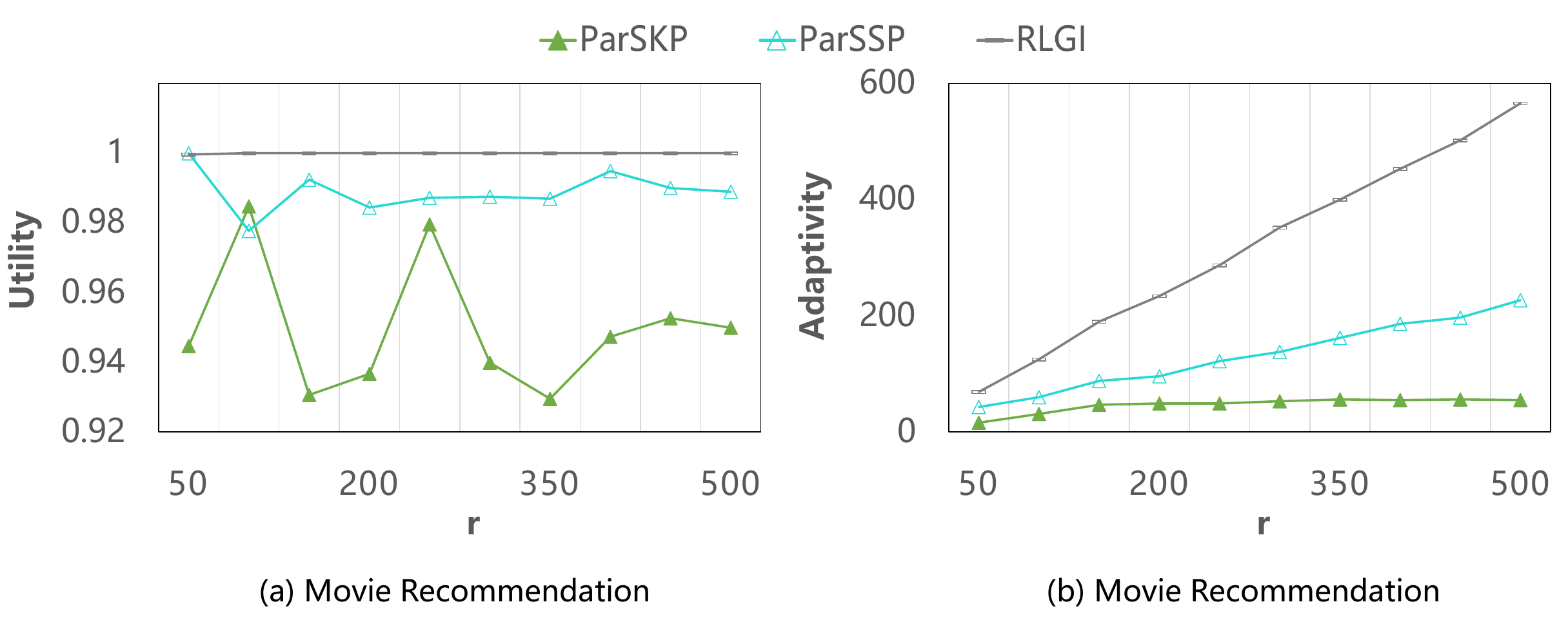}}
		\caption{The figure compares the implemented algorithms on utility and adaptivity, where the plotted utilities are normalized by the largest utility achieved by all algorithms}
		\label{fig:exp-cardinality}
	\end{center}
\end{figure}
\bl{In this section, we test the performance of our \algone~and \PS~algorithms on SMC problems.} The baseline we use is the \textsf{Random Lazy Greedy Improved algorithm} (abbreviated as \textsf{RLGI}) proposed by \shortcite{buchbinder2015comparing}, which is the state-of-the-art non-parallel algorithm with optimal approximation ratio and super-linear adaptivity for the non-monotone SMC. \bl{We conduct this comparison to observe how closely our algorithms can match the utility of the algorithm with optimal approximation ratios in practice, while also highlighting the efficiency advantages of our algorithms over non-parallel algorithms.} For all the algorithms tested, the accuracy parameter $\epsilon$ is set to $0.1$. Each randomized algorithm is executed independently for $10$ times, and the average result is reported. \bl{Following~\shortcite{amanatidis2023submodular,amanatidis2021submodular}, we employ the USM (Unconstrained Submodular Maximization) algorithm from~\shortcite{feige2011maximizing}, which achieves a 1/4-approximation and $\OO(1)$ adaptivity, as a subroutine of our algorithms. Besides, our algorithms are implemented using binary search.} All experiments are run on a Linux server with Intel Xeon Gold 6126 @ 2.60GHz CPU and 256GB memory. Since our algorithms' contribution to the SMC problem is merely a by-product, we simplify our experiments and test the implemented algorithms on only one application. The specific experimental settings are as follows.
\begin{itemize}
	\item Movie Recommendation. By setting $m=\infty$, $\forall g\in G:m_g=\infty$ and $\forall u\in\N:c(u)=1$, we remove the $k$-system constraint and transform the knapsack constraint as the cardinality constraint. Thus this application now is an SMC. In our experiments, We use the MovieLens dataset~\shortcite{badanidiyuru2020submodular,haba2020streaming} which contains 1,793 movies from three genres “Adventure”, “Animation” and “Fantasy”. We scale the maximum cardinality of any feasible solution (i.e., $r$) to compare the performance of all algorithms.
\end{itemize}
\textbf{Experimental Results.} 
In Fig.\ \ref{fig:exp-cardinality}(a), the results show that \algone~(resp. \PS) achieves 93\%-98\% (resp. 98\%-100\%) of the utility of RLGI. In Fig.\ \ref{fig:exp-cardinality}(b), we compare the implemented algorithms on adaptivity, and the results show that \algone~incurs 2-4 times fewer rounds than \PS, and 4-10 times fewer rounds than RLGI. The above experimental results demonstrate that our algorithms can achieve utility performance comparable to that of the algorithm with optimal approximation ratio using much fewer adaptive rounds.
\section{Conclusions}\label{sec:con}
In this paper, we propose \algone, a low-adaptivity algorithm that provides an $(1/8-\epsilon)$ approximation ratio for the problem of non-monotone submodular maximization subject to a knapsack constraint (SKP). To the best of our knowledge, our \algone~algorithm is the first of its kind to achieve either near-optimal $\mathcal{O}(\log n)$ adaptive complexity or near-optimal $\tilde{\OO}(n)$ query complexity for the non-monotone SKP. Furthermore, we propose \PS, a low-adaptivity algorithm that provides an $(1-\epsilon)^{5}(\sqrt{k+1}+1)^{-2}$ approximation ratio for the problem of non-monotone submodular maximization subject to a $k$-system constraint (SSP). Our \PS~algorithm is the first of its kind to achieve sublinear adaptive complexity for the non-monotone SSP. Additionally, we demonstrate that our two algorithms can be extended to solve the problem of submodular maximization subject to a cardinality constraint, achieving performance bounds that are comparable to those of existing state-of-the-art algorithms. During our literature review, we identified and discussed theoretical analysis errors presented in several related studies. Finally, we have validated the effectiveness of our algorithms through extensive experimentation on real-world applications, including revenue maximization, movie recommendation, and image summarization.

\section*{Acknowledgments}
This work is partially supported by 
the National Natural Science Foundation of China (NSFC) under Grant No. 62172384, and the Alibaba Group through Alibaba Innovative Research Program.

\appendix
\newcommand{\cH}{\ensuremath{\mathcal{H}}\xspace}
\section{A Subtle Issue in \shortcite[Theorem 1]{amanatidis2021submodular}}\label{app:amanatidis_error}

In Theorem~1 of \shortcite{amanatidis2021submodular}, an approximation ratio of $(3-\sqrt{3})/12-\Theta(\epsilon)$ is proved for their Algorithm~3 (i.e., \textsf{ParKnapsack}). In their Algorithm~3, they first randomly delete each element from the ground set $\N$ with probability of $1-p$ to get a new ground set $H$, and then use $H$ to run a procedure \TE for multiple times, and finally return $S$, which is the result of one run of \TE. Note that there are two randomnesses in their algorithm: the randomness for generating $H$, and the randomness for generating $S$ by \TE~given a fixed $H$. The derivation of their approximation ratio is based on the analysis of two events: 
\begin{itemize}
	\item $\E[c(S)]<(1-\hat{\varepsilon})\tfrac{B}{2}$ given a fixed $H$. For convenience, let us denote this event by $\cH_{<}$ and put it in a more clear form, i.e., $\cH_{<}=\{ \E_S[c(S)\mid H]< (1-\hat{\varepsilon})\tfrac{B}{2} \}$.
	\item $\E[c(S)]\geq(1-\hat{\varepsilon})\tfrac{B}{2}$ given a fixed $H$. For convenience, let us denote this event by $\cH_{\geq}$ and put it in a more clear form, i.e., $\cH_{\geq}=\{ \E_S[c(S)\mid H]\geq (1-\hat{\varepsilon})\tfrac{B}{2} \}$.
\end{itemize}
Note that the distribution of $H$ making the event $\cH_{<}$ (or $\cH_{\geq}$) happen is \textbf{different} from the original distribution of $H$ where each element in $\N$ appears in $H$ with a probability of $p$. In a nutshell, the $(3-\sqrt{3})/12-\Theta(\epsilon)$-ratio of \shortcite{amanatidis2021submodular} is derived by: (Step I) Deriving $\E[f(S\cup O_H)]\geq p(1-p)f(O)$; (Step II) Deriving an upper bound of $\E[f(S\cup O_H)]$ using $ALG$;  and (Step III) Using Step I and Step II to bridge $ALG$ and $f(O)$, where $O_H=O\cap H$. The subtle problem lying in their analysis is that, the two $\E[f(S\cup O_H)]$'s they use in Step I and Step II are actually different: the one in Step I considers all randomness and hence the original distribution of $H$, while the one in Step II is actually conditioned on $\cH_{<}$ and hence only considers a biased distribution of $H$. Therefore, Step III cannot be done. In the sequel, we explain this in more detail.

In the proof of Theorem 1 by \shortcite[pp.\ 6--7]{amanatidis2021submodular} (see their Section 3), $f(S\cup O_H)$ is used to build the connection between $OPT$ and $ALG$. When analyzing the case $\E[c(S)]<(1-\hat{\varepsilon})\frac{B}{2}$ given a fixed $H$ (i.e., event $\cH_{<}$ occurs), a more careful analysis, via their Lemmata 5 and 6, is conducted (starting from the line just below Eq.\ (6) in the right column of page 6). Then, it shows that keeping fixed $H$,
\begin{equation*}
\E[f(S\cup O_H)\mid \cG]\leq (1+\hat{\varepsilon}+q)ALG+\tau c(O_H)-q\tau\tfrac{B}{2},\tag{first equation in their page 7; (a)}
\end{equation*}
where the event \cG is defined such that at least one of the parallel runs of \TE~outputs $S$ with $c(S)<\frac{B}{2}$ and that solution is considered. Subsequently, it shows that, ``move on to the expectation with respect to $H$'',
\begin{align*}
\E[f(S\cup O_H)]
&=\E[f(S\cup O_H)\mid \cG]\bP(\cG)+\E[f(S\cup O_H)\mid \cG^C]\bP(\cG^C)\\
&\leq (1+\hat{\varepsilon}+q)ALG+\tau pc(O)-q\tau\tfrac{B}{2}+2\hat{\varepsilon}f(O).\tag{second equation in their page 7; (b)}
\end{align*}
Combing that with their Equation (5), which is copied as follows
\begin{equation*}
p(1-p)f(O)\leq \E[f(S\cup O_H)],\tag{their Equation (5) in page 6; (c)}
\end{equation*}
they finally obtain
\begin{equation*}
f(O)\leq \frac{1+q+\hat{\varepsilon}}{p(1-p)-\alpha p+\frac{\alpha q }{2}-2\hat{\varepsilon}}ALG. \tag{their Equation (7) in page 7; (d)}
\end{equation*}

Unfortunately, there is a gap in this claim by abusing the expectations. Because Eqn. (a) and Eqn. (b) listed above only hold when event $\cH_{<}$ occurs according to their reasoning, while from Eqn.(a) to Eqn. (b) they never ``move on to the expectation with respect to $H$'' as they never bound $\E[f(S\cup O_H)]$ under the event $\cH_{\geq}$, nor $\bP(\cH_{<})$ and $\bP(\cH_{\geq})$. To explain this more clearly, we give the full expression of the above equations. Recall that Eqn. (a) and Eqn. (b) only hold when $\mathcal{H}_<$ occurs. Then,
\begin{align}
&\E_H\Bigl[\E_S[f(S\cup O_H) \mid H]\Bigm\vert \mathcal{H}_<\Bigr]\nonumber\\
&=\E_H\Bigl[\E_S[f(S\cup O_H) \!\mid\! H, \cG]\bP(\cG\!\mid\! H)\!+\!\E_S[f(S\cup O_H) \!\mid\! H,\cG^C]\bP(\cG^C\!\mid\! H)\!\Bigm\vert\! \mathcal{H}_<\Bigr]\nonumber\\
&\leq \E_H\Bigl[(1+\hat{\varepsilon}+q)ALG+\tau c(O_H)-q\tau\tfrac{B}{2}+2\hat{\varepsilon} f(O)\Bigm\vert \mathcal{H}_<\Bigr].\label{eq:H<}
\end{align}
On the other hand, their Equation (5) (i.e. Equation (c) listed above) can be fully expressed as
\begin{equation}\label{eq:H}
p(1-p)f(O)\leq \E_{H}\Bigl[\E_{S}[f(S\cup O_H) \mid H]\Bigr]=\E_{H}\Bigl[\E_{S}[f(S\cup O_H \mid H)]\Bigm\vert \cH_{\geq}\cup\cH_{<}\Bigr].
\end{equation}
Note that the correctness of Equation~\eqref{eq:H} is ensured by the fact that elements belong to $H$ with probability $p$. However, when $\cH_{<}$ occurs, it is not guaranteed that the \textit{conditional} probability of every element belonging to $H$ is $p$ (or upper bounded by $p$). As a result, it is invalid to claim
\begin{equation*}
p(1-p)f(O)\stackrel{?}{\leq } \E_H\Bigl[\E_S[f(S\cup O_H) \mid H]\Bigm\vert \mathcal{H}_<\Bigr].\tag{invalid claim}
\end{equation*}
Putting it together, we cannot directly establish the relation between $p(1-p)f(O)$ and $\E_H[(1+\hat{\varepsilon}+q)ALG+\tau c(O_H)-q\tau\tfrac{B}{2}+2\hat{\varepsilon} f(O)\mid \mathcal{H}_<]$ using Equations~\eqref{eq:H<} and \eqref{eq:H}. Therefore, the claim of their Equation~(7) is invalid.

Another similar issue applies to the use of $\E[O_H]=pc(O)\leq pB$ (in their page 7). Note that in Equation~\eqref{eq:H<}, what we actually need is $\E_H[c(O_H)\mid \mathcal{H}_<]$. However, in general, it is trivial to see that
\begin{equation*}
\E_H[c(O_H)\mid \mathcal{H}_<]\neq \E_H[c(O_H)]=pc(O),
\end{equation*}
since when $\cH_{<}$ occurs, it is not guaranteed that the conditional probability of every element belonging to $H$ is $p$.

The same issues as explained above also exist in \shortcite[Theorem 4]{amanatidis2021submodular}, where they claim a $(3-2\sqrt{2}-\epsilon)$-approximation under $\OO(\log n)$ adaptivity for the problem of non-monotone submodular maximization with a cardinality constraint. The reason is that their Theorem~4 uses almost identical methods and analysis as those in their Theorem 1 (see their Appendix D).


\section{Abuse of Markov's Inequality for Non-monotone Submodular Optimization Problems}\label{app:kuhnle_error}
\shortcite{quinzan2021adaptive} recently studied the problem of non-monotone submodular maximization subject to a $k$-system constraint. However, their analysis contains flaws, as explained below. In the proof of \shortcite[Lemma 7]{quinzan2021adaptive} (see page 13 in Appendix of their full version at arXiv:2102.06486v1), the following claim is incorrect:
\begin{equation*}
\E_{a_i}[f(a_i\mid \{a_1,\dotsc,a_{i-1}\})]\geq \Pr[f(a_i\mid\{a_1,\dotsc,a_{i-1}\})>\delta]\delta.
\end{equation*}
The issue arises because $f(a_i\mid\{a_1,\dotsc,a_{i-1}\})$ can be negative since $f(\cdot)$ is non-monotone, which violates the non-negative requirement of Markov’s inequality.


\vskip 0.2in
\bibliography{sample}
\bibliographystyle{theapa}

\end{document}